\documentclass[11pt]{article}
\usepackage[margin=1in]{geometry}
\usepackage[T1]{fontenc}
\usepackage{lmodern}
\usepackage{microtype}
\usepackage{amsmath,amssymb,amsthm,mathtools}
\usepackage{booktabs,array,tabularx}
\usepackage{enumitem}
\usepackage{xcolor}
\usepackage{tikz}
\usetikzlibrary{arrows.meta,positioning,calc,patterns}
\usepackage{hyperref}
\usepackage{cleveref}
\hypersetup{colorlinks=true,linkcolor=blue!55!black,
 citecolor=blue!55!black,urlcolor=blue!55!black,
 pdftitle={Fully Dynamic Unconstrained Submodular Maximization},
 pdfauthor={Mohammad Hajiaghayi},
 bookmarksnumbered=true}
\newtheorem{theorem}{Theorem}[section]
\newtheorem{lemma}[theorem]{Lemma}
\newtheorem{proposition}[theorem]{Proposition}
\newtheorem{corollary}[theorem]{Corollary}
\theoremstyle{definition}

\theoremstyle{remark}

\newcommand{\OPT}{\operatorname{OPT}}
\newcommand{\E}{\mathbb E}
\newcommand{\Poi}{\operatorname{Poi}}

\newcommand{\1}{\mathbf{1}}
\newcommand{\rhosq}{\rho_{\mathrm{sq}}}
\setlist{topsep=4pt,itemsep=2pt,leftmargin=1.6em}
\numberwithin{equation}{section}
\title{Fully Dynamic Unconstrained Submodular Maximization%\\[3pt]
}
\author{Mohammad T. Hajiaghayi\thanks{University of Maryland. Email: \texttt{hajiagha@umd.edu}. The work is partially supported by DARPA expMath, ONR MURI 2024 award on Algorithms, Learning, and Game Theory, Army-Research Laboratory (ARL) grant W911NF2410052, NSF AF:Small grants 2218678, 2114269, 2347322.}}

\date{}
\makeatletter
\renewcommand{\@maketitle}{%
  \newpage\null\vskip 1em%
  \begin{center}%
    {\LARGE\@title\par}%
    \ifx\@author\@empty\else
      \vskip 1.5em%
      {\large\lineskip .5em%
       \begin{tabular}[t]{c}\@author\end{tabular}\par}%
    \fi
  \end{center}%
  \par\vskip .5em}
\makeatother
\begin{document}
\maketitle
\begin{abstract}
We study unconstrained maximization of a fixed nonnegative submodular
function when the ground set changes by insertions and deletions.
A uniformly random subset gives a $1/4$-approximation with constantly
many value-oracle queries per update. We obtain improvements over this
bound for arbitrary oblivious update sequences, without knowing future
deletions.

For value estimation, we give a $(\sqrt2-1-\varepsilon)$-approximation
using $O_{\varepsilon,\delta}(\log^2(N+1))$ expected amortized oracle
queries, with failure probability at most $\delta$ at each fixed time.
Here $N$ bounds the number of active elements. The algorithm maintains
Poisson samples of the positive mass of random chain lower bounds.
We prove that the factor $\sqrt2-1$ is tight for this certificate,
even when polynomially many independent orders are used.

For every fixed integer $K\ge2$, we maintain a feasible
$(1/4+\Omega(K^{-2}))$-approximate set using
$\widetilde O_K(N^{1/K})$ amortized queries.
Thus every fixed exponent $\gamma>0$ permits a fixed improvement
over $1/4$ with $\widetilde O_\gamma(N^\gamma)$ update queries.
In the square-root regime, a stronger algorithm attains
$\rho_{\rm sq}-\varepsilon$, where $\rho_{\rm sq}>0.3177$.
The maintained sets are feasible after every update, and their
approximation guarantees hold in expectation.
All bounds count oracle queries; the value estimator may query
previously seen inactive elements, whereas the set algorithms use
only currently active elements.
\end{abstract}

\section{Introduction}
\label{sec:introduction}

Unconstrained submodular maximization asks for a set of maximum value
under a nonnegative submodular function. The absence of a feasibility
constraint does not make the problem trivial: adding an element can
decrease the objective, and the optimum may contain an arbitrary
number of elements. The problem contains directed and undirected
maximum cut and is a basic test case for nonmonotone optimization.
Feige, Mirrokni, and Vondr\'ak initiated its systematic approximation
study at FOCS~2007~\cite{FMV07,FMV11}. Their work established, among
other results, the $1/4$ guarantee of a uniformly random subset and
the value-oracle hardness of improving upon $1/2$ by a fixed constant with
subexponentially many queries.
The optimal randomized $1/2$-approximation was subsequently obtained
by Buchbinder, Feldman, Naor, and Schwartz through the double-greedy
algorithm~\cite{BFNS12,BFNS15}.\footnote{Their FOCS~2012 paper
received the ten-year Test of Time Award at
FOCS~2022.}

We consider the fully dynamic version of the problem. The function
is fixed, but the available elements change over time. After each
insertion or deletion, the algorithm must approximate
\[
 \OPT_t=\max_{S\subseteq U_t}f(S),
\]
where $U_t$ is the current active universe. We assume an arbitrary
oblivious update sequence: it is fixed independently of the
algorithm's random choices. In particular, insertions and deletions
may be interleaved without restriction, and deletion times are not
announced in advance.

Dynamic submodular maximization has received considerable attention
under cardinality, matroid, matching, and related constraints.
The independent works of Monemizadeh~\cite{Mon20} and Lattanzi
et al.~\cite{LMNTZ20} initiated this line for monotone objectives
under a cardinality constraint. Subsequent work obtained
polylogarithmic update bounds for monotone cardinality-constrained
maximization~\cite{BBGHJM23a}, studied richer constraints such as
matroids, matchings, and matchoids~\cite{DFLNZ23,BBGHJM24,BBGHJM25,BGHJM25},
and extended dynamic techniques to nonmonotone
objectives~\cite{BBGHJM23b,BGHJM26,GZZ26}.
For example, Banihashem et al.~\cite{BBGHJM25} maintain a constant
approximation for monotone submodular matching with expected
amortized update time polynomial in $\log N$ and the inverse
accuracy parameter.

A useful source of these algorithms is the connection between
streaming and dynamic optimization. Thresholding and sampling
procedures developed for streaming can often be organized into
levels that tolerate deletions. This connection is particularly
effective when a cardinality or rank parameter bounds the size of
the maintained solution. Unconstrained maximization lacks such a
bound. Even for the modular function $f(S)=|S|$, any constant-factor
solution contains a linear number of elements. Thus a streaming
procedure that explicitly stores a solution in space sublinear in
the ground-set size cannot serve as a general starting point.
Likewise, a dynamic bound polynomial in a cardinality parameter
need not be sublinear in $N$ when the constraint is vacuous.

There is nevertheless an elementary fully dynamic baseline.
Assign an independent fair bit to each active lifetime and retain
the elements whose bits are one. At every fixed time the retained
set is a uniform subset of $U_t$, so it is a $1/4$-approximation in
expectation. One oracle query after each update gives its current
value. At the other extreme, rerunning double greedy gives a
$1/2$-approximation using $O(N)$ queries per update.
The question is whether one can improve upon $1/4$ while paying
substantially less than linear cost for an unrestricted mixture of
insertions and deletions.

Banihashem et al.~\cite{BBG+26} address this question
in several models. On an $n$-element ground set, they give a
$0.30$-approximation with $O(\sqrt n)$ amortized queries for
insertions, the same guarantee for deletions when their order is
known, and a $0.264$-approximation with $O(\sqrt n)$ amortized
queries for mixed updates when each element's deletion time is
revealed upon insertion. Here $n$ bounds the underlying ground set;
our parameter $N$ bounds only the number of simultaneously active elements.
The mixed-update guarantee is a fully dynamic result
with advance deletion information. Our results remove that
information requirement. They also give different approximation
and update bounds depending on whether the output is a value
estimate or a maintained feasible set.
To our knowledge, these are the first improvements over $1/4$
with sublinear amortized oracle complexity for fully dynamic
unconstrained maximization without advance information about
future updates.

\subsection{Our results}
\label{sec:results}

Let $N\ge1$ be a known upper bound on the number of active elements.
The objective is accessed through a value oracle. Throughout, an
update bound counts \emph{oracle queries}; it does not assert an
identical bound on running time for constructing and manipulating
the queried sets. All approximation guarantees are for arbitrary
oblivious update sequences.

Our first result maintains an estimate of the optimum value.
For fixed accuracy and confidence, its expected amortized query
complexity is logarithmic squared in $N$.

\begin{theorem}[Value estimation]
\label{thm:main-value}
For every $0<\varepsilon<1/10$ and $0<\delta<1/2$, there is a
fully dynamic algorithm that reports a number $L_t$ satisfying
\[
 \Pr\{(\sqrt2-1-\varepsilon)\OPT_t\le L_t\le\OPT_t\}
 \ge1-\delta
\]
at each fixed time $t$. Starting with $n_0$ active elements and
processing $T$ updates, its expected total number of oracle
queries is
\[
 O\!\left(1+(n_0+T)\varepsilon^{-3}
       \log\frac2\delta
       \log\frac{2}{\varepsilon\delta}
       \log\frac{N+1}{\varepsilon}\log(N+1)\right).
\]
For fixed $\varepsilon$ and $\delta$, this simplifies to
$O(1+(n_0+T)\log^2(N+1))$ expected queries: after preprocessing,
each insertion or deletion costs $O(\log^2(N+1))$ queries in
expectation, amortized over the update sequence.
The algorithm uses a fixed global value oracle and may query
previously seen identities that are currently inactive.
\end{theorem}

The leading approximation factor is $\sqrt2-1\approx0.41$.
On the stated confidence event, the
estimate is bounded by a modular
lower certificate of some feasible set, but the algorithm does not
maintain a set realizing that certificate. We give a family of
directed cuts on which the underlying uniform-chain certificate
converges to $\sqrt2-1$. This is a limitation of the certificate,
not an upper bound for fully dynamic value estimation.
The probability guarantee is pointwise in time; replacing $\delta$
by $\delta/(H+1)$ gives simultaneous accuracy through any specified
finite horizon $H$.

Our second main result maintains an actual feasible set, with
oracle queries involving only active elements.

\begin{theorem}[Maintained solutions]
\label{thm:main-sets}
There is an absolute constant $c_0>0$ such that, for every
integer $K\ge2$, a fully dynamic algorithm
maintains $S_t\subseteq U_t$ with
\[
 \E f(S_t)\ge\left(\frac14+\frac{c_0}{K^2}\right)\OPT_t.
\]
For an initial universe of size $n_0$ followed by $T$ updates,
the total query bound is
\[
 O\!\left(
 1+n_0K^2\log^4((N+1)K)
 +TK^3N^{1/K}\log^4((N+1)K)\right).
\]
For fixed $K$, the amortized query bound is
$\widetilde O_K(N^{1/K})$.

In particular, for every fixed $0<\varepsilon<1/10$, a fully
dynamic algorithm maintains $S_t\subseteq U_t$ with
\[
 \E f(S_t)\ge(\rho_{\rm sq}-\varepsilon)\OPT_t,
 \qquad \rho_{\rm sq}>0.3177,
\]
using $\widetilde O_\varepsilon(\sqrt N)$ amortized oracle
queries per update.
The query bounds hold for every outcome of the random choices.
Each algorithm stores a feasible output set after every update.
\end{theorem}

Taking $K=\max\{2,\lceil1/\gamma\rceil\}$ gives a fixed improvement
above $1/4$ with $\widetilde O_\gamma(N^\gamma)$ amortized
queries, for every fixed $\gamma>0$.
The square-root algorithm gives a stronger approximation than
the general bound at $K=2$. We develop this algorithm first,
then prove the tradeoff for arbitrary $K$. Taking $K=3$, for
example, gives cube-root update complexity without a separate
construction.

\begin{table}[t]
\centering
\small
\begin{tabularx}{\textwidth}{@{}lXl@{}}
\toprule
Output & Approximation & Amortized oracle queries \\
\midrule
Feasible set & $1/4$ (uniform subset) & $O(1)$ \\
Value estimate & $\sqrt2-1-\varepsilon$ & $O_{\varepsilon,\delta}(\log^2(N+1))$ \\
Feasible set & $\rho_{\rm sq}-\varepsilon$, $\rho_{\rm sq}>0.3177$
  & $\widetilde O_\varepsilon(\sqrt N)$ \\
Feasible set & $1/4+\Omega(K^{-2})$ & $\widetilde O_K(N^{1/K})$ \\
\bottomrule
\end{tabularx}
\caption{Guarantees for arbitrary oblivious mixed updates, with no
advance deletion information. Set guarantees hold in expectation;
the value guarantee holds with probability $1-\delta$ at each
fixed time. The value estimator uses a global oracle that remains
available on past identities.}
\label{tab:results}
\end{table}

The value estimator has polylogarithmic update complexity, while
the maintained-set algorithms permit every fixed positive update
exponent. The hierarchy makes the dependence explicit: increasing
its depth reduces the exponent and gives an advantage over $1/4$
of order $K^{-2}$. The separate treatment of the two output tasks
reflects the distinction between estimating positive certificate
mass and maintaining the coordinates that realize it.

\subsection{Our approach and techniques}
\label{sec:approach}
\label{sec:overview}

Fully dynamic unconstrained submodular maximization faces two obstacles that do not arise together in static or constrained monotone settings. First, without a cardinality or matroid rank bound $k$, optimal solutions can have size $\Theta(N)$ (even for $f(S)=|S|$), so streaming and dynamic coreset summaries of size $\operatorname{poly}(k,\varepsilon^{-1}\log N)$ do not give sublinear update cost; even statically, outputting a feasible $(1/4+\Omega(1))$-approximate set requires $\Omega(N/\log N)$ value queries~\cite{LFKK22}. Second, randomized double greedy~\cite{BFNS12,BFNS15} achieves $1/2$ statically, but its sequential decisions depend on prefix marginals, so one early update can invalidate the subsequent chain. Conversely, a uniform subset $R_t \sim U_t(1/2)$ (where $C(p)$ includes each element of $C$ independently with probability $p$) uses $O(1)$ queries per update and gives $\E f(R_t)\ge \frac{1}{4}\OPT_t$~\cite{FMV07,FMV11}, but ignores marginal structure. To beat $1/4$ in sublinear update cost without advance deletion information, we maintain marginal structure without sequential decision cascades, separating \emph{value estimation} (Theorem~\ref{thm:main-value}) from \emph{maintaining a feasible set} (Theorem~\ref{thm:main-sets}, first for $\widetilde{O}(\sqrt{N})$ queries and then for $\widetilde{O}_K(N^{1/K})$).

\subsubsection{Value estimation in polylogarithmic queries}

The static $\Omega(N/\log N)$ lower bound~\cite{LFKK22} applies to returning a feasible $(1/4+\Omega(1))$-approximate set. To estimate the scalar $\OPT_t$, we instead construct a \emph{modular lower-bound certificate} whose value is at least $(\sqrt{2}-1)\OPT_t$ in expectation and can be tracked by low-rate Poisson sampling.

\paragraph{The random-chain certificate and the $(\sqrt{2}-1)$ endpoint inequality.}
Fix an ordering $\pi$ of the active universe $U$ and write $f_0 = f(\varnothing)$ and $\OPT(U)=\max_{S\subseteq U}f(S)$. For each $i \in U$, let $P_i = \{j \in U : \pi_j < \pi_i\}$ be the full strict prefix preceding $i$, and define the chain marginal and its positive-mass certificate by
\[
  w_i^\pi(U) = f(P_i \cup \{i\}) - f(P_i), \qquad Q_\pi(U) = f_0 + \sum_{i \in U} \big(w_i^\pi(U)\big)_+.
\]
Here $(x)_+=\max\{x,0\}$ and $(x)_-=\max\{-x,0\}$. Unlike a greedy sequence that discards unprofitable elements, $P_i$ retains all predecessors of $i$, including those with $w_j^\pi(U) < 0$. Conditioning on the full prefix gives two key properties: (1)~building any $S \subseteq U$ along $\pi$ places each $i \in S$ after a subset of $P_i$, so submodularity yields Edmonds' modular lower bound~\cite{Edmonds70}, $f(S) \ge f_0 + \sum_{i \in S} w_i^\pi(U)$ (Lemma~\ref{lem:val-modular}), which on $S^+ = \{i \in U : w_i^\pi(U) > 0\}$ gives $Q_\pi(U) \le f(S^+) \le \OPT(U)$ deterministically; and (2)~since $P_i$ is fixed by $\pi$ rather than by prior sign decisions, inserting or deleting $u$ simply adds or removes $u$ in $P_i$ for $\pi_i > \pi_u$, without cascading.

For a \emph{uniformly random} order $\pi$, we prove in Theorem~\ref{thm:val-chain} the endpoint-enhanced bound
\begin{equation}
  \E_\pi\big[Q_\pi(U)\big] \ge \alpha \OPT(U) + \beta\big(f(\varnothing) + f(U)\big), \qquad \alpha = \sqrt{2}-1 \approx 0.414, \quad \beta = \frac{1-\alpha}{2}.
  \label{eq:overview-chain}
\end{equation}
Realizing $\pi$ via independent priorities $\pi_i \sim \operatorname{Unif}[0,1]$ makes $P_i$ given $\pi_i = p$ a Bernoulli-$p$ sample $R_{-i}(p)$ of $U \setminus \{i\}$. Let $F(x) = \E f(R(x))$ be the multilinear extension, where $R(x)$ includes each $i$ independently with probability $x_i$, and let $H(p) = F(p\1)$ be its diagonal restriction; $\1_A$ denotes the indicator vector of $A$ and $\1=\1_U$. Write $f(i\mid A)=f(A\cup\{i\})-f(A)$. Then $H'(p) = G_+(p) - G_-(p)$, where $G_\pm(p) = \sum_{i \in U} \E\big[(\pm f(i \mid R_{-i}(p)))_+\big]$ and $\E_\pi[Q_\pi(U)] = f_0 + \int_0^1 G_+(p)\,dp$. For an optimal set $O \subseteq U$ ($M = \OPT(U) = f(O)$), comparing $p\1$ with $\1_O$ along two directions---forcing $O$ into the prefix ($p\1 \to \1_O + p\1_{U \setminus O}$, nonnegative) and evicting $U \setminus O$ ($p\1 \to p\1_O$, nonpositive)---gives, by directional concavity of $F$, for $0<p<1$,
\[
  G_+(p) \ge M + \frac{pf(U) - H(p)}{1-p}, \qquad G_-(p) \ge M + \frac{(1-p)f_0 - H(p)}{p}.
\]
If $H(p)$ is small, both $G_+(p)$ and $G_-(p)$ are large; if $H(p)$ is large at intermediate $p$, it had to rise from $f_0$ at $p=0$ and must return to $f(U)$ at $p=1$, forcing $H' = G_+ - G_-$ to compensate. Pairing $p$ with $1-p$, we combine the two marginal bounds with nonnegative weights chosen so that integration by parts cancels the expected prefix values. The middle portion of the chain supplies the optimum term, while the positive and negative marginal mass near the ends pays for the boundary term. Splitting at $a = 1-1/\sqrt{2}$ yields~\eqref{eq:overview-chain} with $\alpha = 1-2a = \sqrt{2}-1$, while retaining the endpoint values $f_0$ and $f(U)$. Proposition~\ref{prop:val-sharpness} shows that $\sqrt{2}-1$ is tight for this certificate even on directed bipartite cuts (and under $\operatorname{poly}(N)$ random orders), as evaluating $w_i^\pi$ on the full prefix $P_i$ penalizes $O$ for preceding elements of $U \setminus O$ later omitted.

\paragraph{Sign-monotone updates, product-Poisson transport, and lazy value windows.}
Write $w_i(U)=w_i^\pi(U)$ and $W^+(U)=\sum_{i\in U}(w_i(U))_+$. Although one update near the front of $\pi$ can alter $w_i(U)$ for $\Omega(N)$ elements, submodularity fixes the sign of every change: insertions into a prefix only \emph{decrease} surviving chain marginals, whereas deletions only \emph{increase} them. At scale $A > 0$ and intensity $c > 0$, conditional on the priorities, we maintain independent Poisson counts $K_i \sim \Poi(c\,(w_i(U))_+/A)$ with total $K = \sum_{i \in U} K_i \sim \Poi(c\,W^+(U)/A)$, storing only coordinates with $K_i > 0$ ($O(c)$ elements in expectation when $A = \Theta(\OPT)$). To transport this exact product-Poisson law from snapshot $U_0 = C \dot\cup D$ to $U_1 = C \dot\cup I$ where $C$ is the set of common continuously surviving lifetimes (Lemma~\ref{lem:val-transport} and Figure~\ref{fig:val-transport}), we factor the transition as $U_0 \to C \to U_1$. Under insertions $C \to U_1$, marginals on $C$ only decrease, so we thin existing marks ($K_i > 0$) with probability $(w_i(U_1))_+/(w_i(C))_+ \in [0,1]$ and draw fresh counts on $I$. Under deletions $U_0 \to C$, surviving marginals on $C$ increase by $d_i := w_i(C)-w_i(U_0) \ge 0$, and their positive parts increase by $\Delta_i := (w_i(C))_+ - (w_i(U_0))_+ \in [0,d_i]$. To add these marks without scanning $C$, we exploit telescoping: the sum of chain marginals over a priority interval is the difference of its two bounding prefix values. After querying and storing the old marginals at the deleted positions, we can therefore compute the total drift in any interval using $O(1)$ queries, by subtracting the old interval sum from the new one and adding back the deleted marginals. Let $D_*$ denote the total drift. We draw $\Poi(cD_*/A)$ proposals and locate each proposed coordinate by binary search over interval drift sums. This samples $i$ proportionally to $d_i$ in $O(\log(N+1))$ queries, and accepting with probability $\Delta_i/d_i$ gives the required independent Poisson increments. The total drift is at most $\OPT(C)+\OPT(U_0)$, yielding the expected transport bound of Lemma~\ref{lem:val-transport}. These snapshot queries use the global oracle and may involve previously seen identities that are currently inactive.

To avoid paying $\OPT/A$ at scales $A \ll \OPT$, we restrict each dyadic scale $A = 2^j \ge f_0$ to the \emph{value window} $W_A = \{i \in U : \tau A/N \le v_i \le A\}$ where $v_i = [f(\{i\})-f_0]_+$ and $\tau=\varepsilon/8$ (Section~\ref{sec:val-windows}), maintaining $q = O(\varepsilon^{-1}\log(1/\delta))$ priority replicas at intensity $c = O(\varepsilon^{-2}\log(q/\delta))$. Since $v_i \le A$ on $W_A$, each local update changes $\OPT(W_A)$ by at most $A$: if $m_0 = \OPT(W_A)$ has $z = m_0/A \gg 1$ at a period start, $\OPT(W_A)$ stays above $A$ for $\Omega(z)$ local updates. For a nonempty window, drawing a fresh half-sample $R \sim W_A(1/2)$ and setting $B = \max\{\lambda f(R), f_0, \max_{i \in W_A} v_i\}$, we batch the next $r = \max\{1, \lfloor B/(\kappa A)\rfloor\}$ local updates (for absolute constants $1 \le \lambda < 2\lambda \le \kappa$) before transporting $W_A$. By Lemmas~\ref{lem:val-period}--\ref{lem:val-selection}: (i)~conditional on the history at a period start, the expected batch length is at least a constant multiple of $1+z$, so windows with larger transport costs also wait longer between refreshes. Accounting for all periods begun by the reporting time, including unfinished ones, gives $O_{\varepsilon,\delta}(\log^2(N+1))$ expected amortized queries across the $O(\log(N/\varepsilon))$ windows per element; (ii)~a stale instance can occur only while the current window optimum remains at least twice its scale, so any window at or above its current optimum is always current; and (iii)~when $a_* = \max(\{f_0\}\cup\{v_i:i\in U\})>0$, selecting the first current window $W_A$ with $A\ge a_*$ takes zero oracle queries, captures $(1-2\tau)\OPT(U)$ at scale $A < 2\OPT(U)$, and preserves the conditional product-Poisson law because the schedule uses only fresh half-samples $R$ independent of $\pi$ and the Poisson draws (if $a_*=0$, we report zero). Taking the largest conservative Poisson estimate over the $q$ replicas gives the pointwise high-probability guarantee of Theorem~\ref{thm:main-value}.

\subsubsection{Maintaining feasible solutions: \texorpdfstring{$\widetilde{O}(\sqrt{N})$}{O(sqrt(N))} warmup and \texorpdfstring{$\widetilde{O}_K(N^{1/K})$}{O(N\^{}(1/K))} tradeoff}

The Poisson sampler above does not yield a feasible set of comparable value: on $f(S)=|S|$ with $|U|=N$, $\OPT=N$, a Poisson sample of rate $c/\OPT$ marks only $O(c)$ elements in expectation, giving expected value $O(c) \ll N$. To maintain $S_t \subseteq U_t$ with $\E f(S_t) \ge (\frac{1}{4}+\Omega(1))\OPT_t$, we instead maintain a mask that omits only a small amount of the positive mass of a chain vector $z$.

\paragraph{Decremental mask maintenance via lower anchors and halving repairs.}
In a \emph{decremental cohort} $C$ that sees only deletions (Section~\ref{set-sec-decremental}, Lemma~\ref{set-lem-mask-maintenance}), every surviving chain coordinate $z_i$ (either $w_i^\to$ in Section~\ref{sec:root-updates} or $\frac{1}{2}(w_i^\to+w_i^\leftarrow)$ in Section~\ref{sec:general-tradeoff}, where $\to$ and $\leftarrow$ denote an order and its reverse) can only \emph{increase} under deletions while the order is fixed. Thus once $i \in C$ has $z_i > 0$, it stays positive until the next rebuild; the remaining task is to detect coordinates that were initially $\le 0$ and have drifted upward across zero.
Fix $0<\xi<1/20$. To avoid scanning $C$ after each deletion, we divide the cohort's lifetime into value epochs (setting $L = \max(\{f_0\}\cup\{f(\{i\}):i\in C\})$ at the epoch start, rebuilding when the current maximum drops below $L/2$, and restricting to $W = \{i \in C : s_i \ge \theta\}$ where $s_i = f(\{i\})-f_0$ and $\theta = \xi L/N$). For each $i \in W$, we store a \emph{lower anchor} $b_i \le z_i$ (initialized to $z_i$ in the chain vector on $W$) and maintain the mask of eligible coordinates with $b_i > 0$ (all coordinates for the forward chain, and the priority intervals below for paired chains). Since $d_i := z_i - b_i \ge 0$ and the total drift over any priority interval takes $O(1)$ prefix/suffix queries, after each deletion inside $W$ we sample $i \in W$ proportional to $d_i$ in $O(\log(N+1))$ queries via binary search, updating $b_i \leftarrow z_i$ (and resetting a failure counter) only when $i$ is \emph{productive}---meaning $b_i \le 0 < z_i$ (newly positive) or $b_i \le -\theta$ and $z_i > b_i/2$ (negative magnitude halved)---and stopping after $\tau = \widetilde{O}(\xi^{-1})$ consecutive unproductive draws, or when the total drift vanishes. Each of the $m$ initial elements can undergo only $O(\log((N+1)/\xi))$ negative-anchor halvings and one positive crossing per value epoch. The singleton cutoff also restricts each element to $O(\log((N+1)/\xi))$ retained epochs, independently of the numerical range of $f$. Together with the sampling cost, this bounds the total number of oracle queries over all deletions by $m\cdot O(\xi^{-1}\log^4((N+1)/\xi))$ deterministically. If $E$ is the positive mass omitted among the eligible coordinates, every missing positive coordinate is productive and contributes at least its positive weight to the total drift. Unproductive drift is controlled by the vector's total positive and negative masses, each at most $\OPT(W)$, together with a small contribution from anchors within $\theta$ of zero. Thus, whenever $E>\xi\OPT(W)$, a draw is productive with probability at least $\xi/3$. Choosing the failure threshold as in Lemma~\ref{set-lem-mask-maintenance} makes the probability of stopping with such an error at most $\xi$, giving $\E[E]\le2\xi\OPT(W)$.

\paragraph{Warmup: $\widetilde{O}(\sqrt{N})$ queries via a decremental core and two forced-base extensions.}
Partitioning the update stream into epochs of $d = \lceil\sqrt{N}\rceil$ updates gives our $\widetilde{O}(\sqrt{N})$ algorithm (Section~\ref{set-sec-sqrt}, Theorem~\ref{set-thm-sqrt-dynamic}). At each epoch start, all active elements form a decremental core $C$ ($|C| \le N$) initialized with a decremental mask on a single random forward chain, and an insertion buffer $D$ starts empty ($|D| \le d$). After initialization, prepaying each new core's lifetime bound $\widetilde{O}(N)$ from the preceding $d = \lceil\sqrt{N}\rceil$ updates costs $\widetilde{O}(\sqrt{N})$ amortized queries per update, and its maintained mask $S \subseteq C$ satisfies the endpoint profile (Theorem~\ref{set-thm-core-profile}):
\begin{equation}
  \E f(S) \ge \alpha\OPT(C) + \beta\big(f(\varnothing) + f(C)\big) - \varepsilon\OPT(C).
  \label{eq:overview-core}
\end{equation}
Since optimizing $D$ separately and taking the union with $S$ can destroy value under nonmonotonicity, whereas $S$ alone misses $Y = O \cap D$ (where $O=X\cup Y$ is optimal on $U=C\dot\cup D$ and $X=O\cap C$), we use $O(|D|+1) = O(\sqrt{N})$ queries after each update to run randomized double greedy~\cite{BFNS15} on $D$ with a forced base $L \subseteq C$ (optimizing over $[L, L \cup D]$). Replacing $gh \le \frac{1}{2}(g^2+h^2)$ in the double-greedy proof with $gh \le \theta g^2 + \phi h^2$ when both lower- and upper-endpoint gains $g,h$ are nonnegative ($\theta,\phi \ge 0$, $\theta\phi \ge 1/4$) gives the \emph{endpoint-weighted double-greedy inequality} (Lemma~\ref{lem:prelim-rdg}), $(1 + \theta + \phi)\E f(T) \ge f(Q) + \theta f(L) + \phi f(L \cup D)$ for all $L\subseteq Q\subseteq L\cup D$. We run double greedy on $D$ twice from two different bases in $C$ (Lemma~\ref{set-lem-augmentation}):
(1)~from the maintained core set ($L = S \leadsto T_S$, comparator $Q = S \cup Y$), where submodularity links $S$ to $Y$ through $f(C \cup Y)$ while preserving the contribution of $S$, at the cost of a penalty proportional to $f(C)$; and
(2)~from an independent Bernoulli-$p$ sample ($L = R \sim C(p) \leadsto T_R$, comparator $Q = R \cup Y$), where two-block concavity on $X = O \cap C$ and $Z = C \setminus O$ lower-bounds the expected values of $R$, $R \cup Y$, and $R \cup D$ by the corresponding four-corner averages, in which neither, one, or both core blocks are present (Lemma~\ref{lem:prelim-sampling}).
In a nonnegative linear combination of $\E f(S)$, $\E f(T_S)$, and $\E f(T_R)$, submodularity combines the $f(C \cup Y)$ term from the maintained-core extension with the four-corner bounds from the sampled-core extension. This gives a positive contribution from $f(O)$, together with a loss proportional to $f(X)$ and the penalty on $f(C)$. The $\alpha\OPT(C)$ term in~\eqref{eq:overview-core} absorbs the loss on $f(X)$, while the retained $\beta f(C)$ endpoint term and the sampled-base contribution offset the penalty on $f(C)$. The analysis weights are chosen so that all three output coefficients and the remaining endpoint coefficients are nonnegative (Lemma~\ref{set-lem-augmentation}). After normalization, the best of the three candidates dominates their convex combination. Returning the best of $\{S,T_R,T_S\}$ and choosing the sampling probability and endpoint weights as in Section~\ref{set-sec-root-parameters} yields approximation ratio $\rhosq-\varepsilon$, where $\rhosq>0.3177$.

\paragraph{The general $\widetilde{O}_K(N^{1/K})$ hierarchy and the half-sampling slack $\eta$.}
To achieve update exponent $1/K$ for any integer $K \ge 2$, we replace the core-buffer split with a hierarchy of $K$ disjoint cohorts $C_0, \dots, C_{K-1}$ partitioning $U_t$, with geometrically spaced rebuild periods $p_j = r^{K-1-j}$ for $0\le j\le K-2$, where $r = \lceil N^{1/K}\rceil$ (Section~\ref{set-sec-clean-dynamic-schedule}). Insertions enter $C_{K-1}$ (size $\le r$, rebuilt after each update). At multiples of $p_j$, the highest due level $j$ merges $C_j \cup \dots \cup C_{K-1}$ into a fresh decremental cohort $C_j$ of initial size at most $p_{j-1}=rp_j$ for $j\ge1$, and at most $N\le rp_0$ for $j=0$, and empties the lower slots; after initialization, prepaying its decremental cost $\widetilde{O}(r p_j)$ over the preceding $p_j$ updates charges $\widetilde{O}_K(N^{1/K})$ queries per update per level.
For $K \ge 3$, we cannot rerun double greedy across the hierarchy after each update because intermediate cohorts $C_1, \dots, C_{K-2}$ can have sizes of order $N^{1-j/K} \gg N^{1/K}$ and reset at different frequencies. Each cohort $C_j$ must therefore maintain its decremental mask $S_j \subseteq C_j$ independently, raising two obstacles: (1)~\emph{within each cohort}, a single forward chain gives only a $(\sqrt{2}-1)\OPT(C_j)$ certificate guarantee, which is too small to survive $K$-way composition above $1/4$; and (2)~\emph{across cohorts}, taking the best cohort loses a factor of $K$, whereas including each cohort independently with probability $1/2$ preserves cross-cohort interactions only at rate $\frac{1}{2}(1-\frac{1}{2}) = 1/4$.

We resolve both obstacles by parameterizing the analysis by the \emph{half-sampling slack} of the active universe $U$ with optimum $O$ ($M = f(O) = \OPT(U)$) and $B = U \setminus O$ (Section~\ref{set-sec-clean-structure}),
\[
  \eta := 4\E f(U(1/2)) - \OPT(U) - f_0 - f(U) \ge f(B) \ge 0, \qquad \E f(R_t) \ge \frac{1}{4}M + \frac{1}{4}\eta,
\]
where $R_t \sim U_t(1/2)$ is a maintained uniform half-sample. If $\eta \ge \Omega(K^{-2})M$, $R_t$ already achieves $(\frac{1}{4}+\Omega(K^{-2}))M$, so we may restrict to the \emph{near-tight regime} $\eta \ll K^{-2}M$. The four-corner slack controls the conditional midpoint gaps on $O$ and on $B$, so small $\eta$ makes $f$ nearly modular within each block:
(i)~the \emph{wrong-role budgets} (Section~\ref{set-sec-wrong-role-budgets}) satisfy $r_B := \max_{T \subseteq B} f(T) - f_0 \le 2\eta$ and $r_k := \max_{T \subseteq O} f(B \cup T) - f(U) \le 2\eta$; and
(ii)~for any submodular function, the error in replacing its value by a greedy modular lower bound is nonnegative and at most twice its \emph{midpoint gap}: twice its expected half-sample value minus the sum of its empty- and full-set values. This gap also controls the total expected absolute error in random-prefix marginals, with a loss proportional to $1/\min\{p,1-p\}$ at sampling rate $0<p<1$ (Equation~\eqref{set-eq-clean-residual-gradient}). The complementary conditional functions on $O$ obtained by fixing $Z\subseteq B$ and $B\setminus Z$ have total midpoint gap $O(\eta)$. Together with the wrong-role budgets, this control survives arbitrary restriction and conditioning: for disjoint $R,C\subseteq U$, the centered half-sampling slack of $g(T)=f(R\cup T)$ on $C$ is at most $C_{\mathrm c}\eta=20\eta$ (Theorem~\ref{set-thm-clean-conditional-slack}), where centering means subtracting the minimum value of $g$.

\paragraph{Two shifted paired-chain masks for near-optimal cohort solutions.}
Consider any subuniverse $C \subseteq U$ fixed independently of the priorities (with local optimum $O_C$, $B_C = C \setminus O_C$, and centered slack $\eta_C - m_C \le C_{\mathrm{c}}\eta_U$). Here $\eta_C$ is the half-sampling slack on $C$, $m_C=\min_{T\subseteq C}f(T)$, and $\eta_U=\eta$. In a single forward chain $w_i^\to$ on $C$, an element $i \in O_C$ at priority $\pi_i = p$ is conditioned on a $p$-fraction of $B_C$, whose negative cross-terms can drive $w_i^\to$ to zero even when $\eta_C = 0$. In the reverse chain $w_i^\leftarrow$, $i$ is conditioned on the complementary prefix $(C \setminus \{i\}) \setminus P_i$, containing the remaining $(1-p)$-fraction of $B_C$. We therefore use the \emph{paired chain vector} $z_i = \frac{1}{2}(w_i^\to + w_i^\leftarrow)$ (Section~\ref{set-sec-clean-paired-chains}). Conditioned on $\pi_i = p$, the forward and reverse prefixes on $B_C$ form a complementary pair $(Z, B_C \setminus Z)$ with $Z \sim B_C(p)$. Using the complementary-gap bound, we compare the sum of the two conditional greedy bases with the endpoint bases obtained when none or all of $B_C$ is present, always using a common fixed order on $O_C$ (Equation~\eqref{set-eq-clean-paired-lone}). This comparison, and its symmetric version on $B_C$, gives fixed nonnegative weights for omitted good elements and included bad elements that upper-bound their total loss. Each such loss is charged to a wrong-role budget, a residual from the comparison, or a sign error of the maintained mask. The expected total residual on either block at priority $p$ is $O(\eta_C/\min\{p,1-p\})$ (Equation~\eqref{set-eq-clean-coordinate-error}).
Because integrating $1/\min\{p,1-p\}$ over $(0,1)$ diverges near $p \in \{0,1\}$, we restrict eligibility in $\pi$ to $[1/4, 3/4)$ and maintain a second paired-chain mask under the half-shifted priorities $\pi'_i = (\pi_i + 1/2) \bmod 1$ on $[1/4, 3/4)$. Since the two eligibility regions partition the original priority range $[0,1)$, every element of $C$ is interior in exactly one order. Applying the same argument to $f-m_C$ leaves all marginals and masks unchanged and replaces $\eta_C$ by $\eta_C-m_C$. Thus, after the singleton truncation above, the union $S_j = \widetilde{A}_1 \cup \widetilde{A}_2$ of the two maintained masks in cohort $C_j$ satisfies $\E f(S_j) \ge \OPT(C_j) - D_0 \eta_U - C_{\mathrm{e}}\xi \OPT(U)$ (Corollary~\ref{set-cor-cohort}), where $D_0=20(24+16\ln2)$ and $C_{\mathrm e}=10$.

\paragraph{Balanced $\lfloor K/2 \rfloor$-cohort composition via coarse bipartite quadratics.}
To combine the maintained cohort sets $S_0, \dots, S_{K-1}$ with local deficits $\delta_j := \OPT(C_j) - f(S_j) \ge 0$ (Section~\ref{set-sec-clean-balanced-cohorts}, Theorem~\ref{set-thm-balanced}), write $[K] = \{0,\dots,K-1\}$ and split each $S_j$ into $G_j = S_j \cap O$ and $H_j = S_j \cap B$ ($G = \bigcup_j G_j$, $H = \bigcup_j H_j$), where subadditivity and $r_B \le 2\eta$ give $f(G) \ge M - \sum_j \delta_j - O(K\eta)$. Viewing the coarse groups $G_j$ and $H_j$ as single coordinates, the $O(\eta)$ conditional midpoint gaps on $O$ and $B$ imply that for all $X, Y \subseteq [K]$, with $G_X=\bigcup_{i\in X}G_i$ and $H_Y=\bigcup_{j\in Y}H_j$, $f(G_X \cup H_Y)$ is lower-bounded up to $8K\eta$ error by a \emph{coarse bipartite quadratic polynomial} (Lemma~\ref{set-lem-coarse-lower}):
\[
  Q(X, Y) = f_0 + \sum_{i \in X} a_i + \sum_{j \in Y} b_j - \sum_{i \in X,\, j \in Y} c_{ij}
          = f_0 + \sum_{i \in X} u_i + \sum_{j \in Y} b_j + \sum_{i \in X,\, j \notin Y} c_{ij},
\]
where $a_i,b_j$ are fixed-order chain marginals of the $G_i$'s and $H_j$'s, respectively, and $c_{ij}\ge0$ is the decrease in $G_i$'s chain marginal caused by adding $H_j$ after its predecessors; $u_i = a_i - \sum_j c_{ij}$ satisfies $\sum_i (u_i)_- \le r_k \le 2\eta$, and $Q([K],\varnothing) = f(G)$. Selecting a subset of cohorts $J \subseteq [K]$ and outputting $S_J = \bigcup_{j \in J} S_j = G_J \cup H_J$ separates the diagonal and off-diagonal terms in $Q(J,J)$:
(1)~a \emph{same-cohort (diagonal) term} $c_{jj}$ is never captured, since a cohort's good and bad portions are selected together. Discarding $H_j$ can improve $S_j$ by at most $\delta_j$, while the contribution of $H_j$ alone beyond $f_0$ is at most $r_B$. These two facts, with an $O(\eta)$ correction for conditioning, charge the diagonal interactions to the local deficits, giving $\sum_j c_{jj}\le\sum_j\delta_j+O(K\eta)$ (Equation~\eqref{set-eq-clean-diagonal-budget}); and
(2)~for \emph{cross-cohort (off-diagonal) terms} $c_{ij}$ ($i \ne j$), choosing $J \subseteq [K]$ to be a \emph{uniformly random subset of fixed size $s = \lfloor K/2 \rfloor$}, drawn once independently of the modules and reused throughout the stream, induces negative correlation between distinct cohorts and captures every off-diagonal interaction at rate $q_K := \Pr[i \in J,\, j \notin J] = \frac{\lfloor K/2 \rfloor \lceil K/2 \rceil}{K(K-1)} \ge \frac{1}{4} + \frac{1}{4K}$.
The local deficits are paid twice: once for the good mass omitted by the masks, and once for the diagonal interactions that whole-cohort selection cannot capture. By the cohort guarantee, their total expected cost is $O(K\eta)+O(K\xi M)$. Choosing $\xi=1/(80K^2)$ leaves a gain of order $M/K$ from balanced sampling, up to the slack loss $O(K\eta)$; hence $\E f(S_J)\ge\big(\frac14+\Omega(K^{-1})\big)M-O(K)\eta$. Returning the better of $S_J$ and $R_t$ (where $\E f(R_t) \ge \frac{1}{4}M + \frac{1}{4}\eta$) takes the maximum of two affine bounds in $\eta/M \ge 0$ (the case $M=0$ is immediate) with slopes $-O(K)$ and $+1/4$; their intersection at $\eta/M = \Theta(K^{-2})$ guarantees $\E\max\{f(S_J), f(R_t)\} \ge \big(\frac{1}{4} + \Omega(K^{-2})\big)M$.

\subsection{Related work}
\label{sec:related}

Before double greedy, the approximation guarantee for static
unconstrained submodular maximization was improved through
simulated annealing~\cite{OGV11} and a structural continuous-greedy
approach~\cite{FNS11}. Buchbinder and Feldman later obtained a
deterministic $1/2$-approximation~\cite{BF18}.
Li, Feldman, Kazemi, and Karbasi~\cite{LFKK22} proved that
improving upon $1/4$ by a fixed amount requires
$\Omega(N/\log N)$ queries in the static value-oracle model
when the algorithm must return a feasible set.
Our update bounds are consistent with this lower bound: building
the dynamic structure from an initially empty universe takes
$N$ insertions. They exploit information retained between
successive instances.

For constrained dynamic maximization, the streaming threshold
method of Badanidiyuru et al.~\cite{BMKK14} underlies several
early algorithms. Chen and Peng~\cite{CP22} established
approximation-dependent lower bounds for dynamic monotone
submodular maximization under a cardinality constraint.
Later developments include dynamic matroid
optimization~\cite{DFLNZ23,BBGHJM24}, submodular
matching~\cite{BBGHJM25}, and matchoid constraints~\cite{BGHJM25}.
Dynamic nonmonotone algorithms under cardinality and matroid constraints
have also improved substantially~\cite{BBGHJM23b,BGHJM26,GZZ26}.
In particular, the $0.262$ and $0.277$ guarantees
of Banihashem et al.~\cite{BGHJM26} exceed $1/4$, but their
update bounds have polynomial dependence on the cardinality
parameter. Setting that parameter to $N$ does not give the
sublinear guarantees studied here.
Algorithms with known deletion order~\cite{PR23} and algorithms
with predictions about future updates~\cite{LS24,AB24} study
related models in which advance information can reduce update costs.

The closest work for general unconstrained maximization is
due to Banihashem et al.~\cite{BBG+26}, discussed above.
Its use of a core, an insertion buffer, and forced-set extensions
motivates the corresponding part of our construction.
Our square-root algorithm maintains its core under arbitrary deletions,
strengthens the extension analysis using endpoint terms, and
requires no deletion predictions.
The polylogarithmic value estimator and the arbitrary fixed-exponent
set tradeoff use the chain-sampling and cohort arguments described
above.

\paragraph{Organization.}
Section~\ref{sec:value} develops the value estimator, proves
Theorem~\ref{thm:main-value}, and identifies the limitation of its
uniform-chain certificate. Section~\ref{sec:root-updates}
gives the square-root algorithm. Section~\ref{sec:general-tradeoff}
proves the structural lemmas underlying the cohort algorithms and develops
the composition argument for the general tradeoff in
Theorem~\ref{thm:main-sets}.

\section{Fully dynamic value estimation}
\label{sec:value}

We first maintain a $(\sqrt2-1-\varepsilon)$ lower estimate of the
optimum with polylogarithmic expected amortized oracle queries.
The estimator tracks the positive mass of a random-chain modular
lower bound. We prove the chain inequality, give the fully dynamic
maintenance algorithm, and then identify the limitation of the
certificate on which the estimate is based.

\subsection{Model and preliminaries}
\label{sec:preliminaries}

\paragraph{Model and notation.}
Let $f:2^V\to\mathbb R_{\ge0}$ be a fixed submodular function.
We do not assume $f(\varnothing)=0$.
For disjoint $A,B$, write $f(B\mid A)=f(A\cup B)-f(A)$,
and use $f(i\mid A)$ when $B=\{i\}$.
For $C\subseteq V$, put
\[
 \OPT(C)=\max_{S\subseteq C}f(S),\qquad f_0=f(\varnothing).
\]
An update inserts or deletes one identity; an identity has at
most one active copy. Reinsertions are permitted, and each
insertion starts a new lifetime. The upper bound $N$ concerns
simultaneously active identities, not the total number that
ever appear. All probability and expectation statements refer
to the algorithm's random choices, with the update sequence fixed.

The value-oracle model charges one query for evaluating $f(A)$.
It does not charge the cardinality of $A$ as additional oracle
queries. Our value estimator uses the global oracle on current
and historical snapshots. Our maintained-set algorithms query
only subsets of the current active universe.
The algorithms use exact oracle answers and real arithmetic.
The quantity $\widetilde O$ suppresses logarithmic factors in
$N$ and in the displayed accuracy parameters. Unspecified
logarithms are natural.

A \emph{stored output set} is represented by maintained
membership lists, together with its current value when candidates
are compared. Returning its representation uses no additional
oracle queries. Enumerating its elements takes ordinary time
proportional to the output length. Neither this enumeration cost
nor the cost of constructing arbitrary queried subsets is
claimed to be polylogarithmic.

\paragraph{Sampling and the multilinear extension.}
For a finite universe, the multilinear extension is
$F(x)=\E f(R(x))$, where the elements of $R(x)$ are chosen
independently with their respective probabilities $x_i$.
Submodularity implies that its mixed second derivatives are
nonpositive. Hence $F$ is concave along any line whose direction
has all coordinates nonnegative or all coordinates nonpositive.
This standard consequence~\cite{CCPV11} will be used repeatedly.

\begin{lemma}[Two-block sampling]
\label{lem:prelim-sampling}
Let $A,B,C$ be disjoint. Independently sample each element
of $A$ with probability $p$ and each element of $B$ with
probability $q$, obtaining $R_A,R_B$. Then
\begin{align*}
 \E f(C\cup R_A\cup R_B)
 \ge{}&(1-p)(1-q)f(C)+p(1-q)f(C\cup A)\\
 &+(1-p)qf(C\cup B)+pqf(C\cup A\cup B).
\end{align*}
In particular, for a uniform half-subset $R\subseteq U$ and
an optimum $O\subseteq U$,
\[
 4\E f(R)\ge
 f(\varnothing)+f(O)+f(U\setminus O)+f(U).
\]
\end{lemma}
\begin{proof}
Apply concavity first along the coordinates of $A$ and then
along the coordinates of $B$, comparing each intermediate
point with the two endpoints of its line segment.
The last assertion takes $C=\varnothing$, $A=O$,
$B=U\setminus O$, and $p=q=1/2$.
\end{proof}

\begin{lemma}[Singleton bounds]
\label{lem:prelim-singletons}
Put $v_i=[f(\{i\})-f_0]_+$. For any $C,D\subseteq V$,
\[
 \OPT(C)\le f_0+\sum_{i\in C}v_i,\qquad
 \OPT(C)\le\OPT(D)+\sum_{i\in C\setminus D}v_i.
\]
Elements with $v_i=0$ may be excluded without decreasing the
optimum. Inserting or deleting an element with $v_i\le A$
changes the optimum by at most $A$.
\end{lemma}
\begin{proof}
Let $S$ be an optimum in $C$. Start from $S\cap D$ and add
the elements of $S\setminus D$, bounding each marginal
by its positive singleton gain. Since $f(S\cap D)\le\OPT(D)$,
this proves the second inequality.
Taking $D=\varnothing$ proves the first
inequality. If $v_i=0$, diminishing returns makes every
marginal of $i$ nonpositive. The final assertion follows by
applying the second inequality in both directions.
\end{proof}

\subsection{The random-chain inequality}
\label{sec:val-chain}

Our value estimator is based on the classical chain modular lower
bound~\cite{Edmonds70}. Choosing a random permutation and taking
its positive support is the random-permutation algorithm of
Iyer, Jegelka, and Bilmes~\cite{IJB13}.
We prove a sharper bound on its certificate, retaining endpoint
terms needed later, and show how to maintain sampled certificate mass.
The chain coordinates may change at many elements
after one update.  The useful property is that these changes all have
the same sign.  This allows us to maintain a sample from the positive
coordinates even when maintaining their complete support would be
expensive.

Fix a finite active universe $U$ and an ordering $\pi$.  Write
$f_0=f(\varnothing)$, let $P_i$ be the strict prefix preceding $i$,
and define
\begin{equation}
 w_i^\pi(U)=f(P_i\cup\{i\})-f(P_i),\qquad
 Q_\pi(U)=f_0+\sum_{i\in U}(w_i^\pi(U))_+.
 \label{eq:val-chain-definition}
\end{equation}
The prefix contains all preceding elements, including those with
negative marginal gain.  This is the distinction between the chain
in~\eqref{eq:val-chain-definition} and a greedy construction that
discards unprofitable elements as it proceeds.

\begin{lemma}[Modular lower bound]
\label{lem:val-modular}
For every order $\pi$ and every $S\subseteq U$,
\[
 f(S)\ge f_0+\sum_{i\in S}w_i^\pi(U).
\]
Consequently,
\[
 f_0\le Q_\pi(U)\le
 f\bigl(\{i:w_i^\pi(U)>0\}\bigr)\le\OPT(U).
\]
\end{lemma}
\begin{proof}
Build $S$ in its inherited order.  When an element $i\in S$ is added,
its predecessors in this construction form a subset of $P_i$.
Its marginal gain is therefore at least $w_i^\pi(U)$.
Summing proves the first inequality; applying it to the positive
coordinates proves the remaining assertions.
\end{proof}

An arbitrary ordering need not give a useful approximation.  A
uniform ordering, however, gives a constant larger than $1/4$.
We retain the endpoint terms in the following theorem because they
will also be used in the construction that maintains a feasible set.
Throughout the paper, put
\[
 \alpha=\sqrt2-1,\qquad \beta=\frac{1-\alpha}{2}.
\]

\begin{theorem}[Random-chain certificate]
\label{thm:val-chain}
For every nonnegative submodular function on $U$, a uniformly random
order satisfies
\begin{equation}
 \E Q_\pi(U)\ge
 \alpha\OPT(U)+\beta\bigl(f(\varnothing)+f(U)\bigr).
 \label{eq:val-chain-profile}
\end{equation}
\end{theorem}

The proof compares a random prefix with an optimum in two ways.
One comparison forces the optimum into the prefix; the other removes
the elements outside the optimum.  If the prefix has little value,
these comparisons require either substantial positive marginal mass
or substantial negative marginal mass.  Pairing the beginning and
the end of the chain lets us combine the two requirements.  The
integration below cancels the expected prefix value and leaves the
factor $\alpha$.

\begin{proof}
Let $M=\OPT(U)$, $f_1=f(U)$, and let $F$ be the multilinear
extension.  For $p\in[0,1]$, define
\[
 H(p)=F(p\mathbf1),\qquad
 G_+(p)=\sum_i\E\bigl[(f(i\mid R_{-i}(p)))_+\bigr],\qquad
 G_-(p)=\sum_i\E\bigl[(-f(i\mid R_{-i}(p)))_+\bigr],
\]
where $R_{-i}(p)$ samples each element of $U\setminus\{i\}$
independently with probability $p$.  In particular, the positive part
is taken before the expectation.  The basic identities are
\begin{equation}
 H'=G_+-G_-,\qquad
 Q:=\E Q_\pi(U)=f_0+\int_0^1G_+(p)\,dp.
 \label{eq:val-diagonal-identities}
\end{equation}
The second identity follows by assigning independent uniform
priorities and conditioning on the priority of each element.

Fix an optimum $O$.  Concavity of $F$ along nonnegative directions
gives
\[
 (1-p)\sum_{i\in O}\partial_iF(p\mathbf1)
 \ge F(\mathbf1_O+p\mathbf1_{U\setminus O})-H(p)
 \ge (1-p)M+pf_1-H(p).
\]
Forcing the coordinates outside $O$ to zero gives similarly
\[
 p\sum_{i\notin O}[-\partial_iF(p\mathbf1)]
 \ge F(p\mathbf1_O)-H(p)
 \ge pM+(1-p)f_0-H(p).
\]
Since an expectation of a positive part dominates the corresponding
expectation, we obtain
\begin{equation}
 G_+(p)\ge M+\frac{pf_1-H(p)}{1-p},
 \qquad
 G_-(p)\ge M+\frac{(1-p)f_0-H(p)}p.
 \label{eq:val-forcing}
\end{equation}

Set $E_0=(f_0+f_1)/2$.  For $0\le p\le1/2$, let
\[
 K(p)=\frac{H(p)+H(1-p)}2-E_0,\qquad
 A(p)=\frac{G_+(p)+G_-(1-p)}2,\qquad
 B(p)=\frac{G_-(p)+G_+(1-p)}2.
\]
Then $K(0)=0$, $K'=A-B$, and $A,B\ge0$.
Equations~\eqref{eq:val-diagonal-identities}
and~\eqref{eq:val-forcing} imply
\begin{equation}
 \begin{split}
 A(p)&\ge M-E_0-\frac{K(p)}{1-p},\\
 B(p)&\ge M-E_0-\frac{K(p)}p,\\
 Q&=E_0+\int_0^{1/2}(A(p)+B(p))\,dp.
 \end{split}
 \label{eq:val-symmetrized}
\end{equation}
Let $a=1-1/\sqrt2$ and
$k(p)=(1/2-p)/(p(1-p))$ on $[a,1/2]$.  Direct differentiation yields
\[
 0\le k\le1,\qquad k(a)=1,\qquad k(1/2)=0,\qquad
 k'+\frac{1-k}{1-p}+\frac{1+k}p=0.
\]
The coefficients $1-k$ and $1+k$ are nonnegative, so
\[
 A+B=(1-k)A+(1+k)B+kK'
 \ge2(M-E_0)
   -K\left(\frac{1-k}{1-p}+\frac{1+k}p\right)+kK'.
\]
Integrating over $[a,1/2]$ and integrating by parts gives
\[
 \int_a^{1/2}(A+B)\,dp
 \ge(1-2a)(M-E_0)-K(a).
\]
On $[0,a]$, the inequality $A+B\ge K'$ gives
$\int_0^a(A+B)\,dp\ge K(a)$.  Hence
\[
 Q\ge E_0+(1-2a)(M-E_0)
   =\alpha M+\beta(f_0+f_1).
\]
The centering by $E_0$ is algebraic; it does not assume that
$f-E_0$ is nonnegative.
\end{proof}

\subsection{Dynamic maintenance}
\label{sec:val-dynamic}

We next maintain an estimate of $Q_\pi(U)$ under arbitrary
insertions and deletions.  The construction has two parts.  At a
fixed value scale, Poisson sampling handles changes to the chain.
A separate schedule batches the work at scales much smaller than
the current optimum.  The schedule is independent of the random
orders and their samples.  This independence is what permits us to
select a current instance without changing its distribution.

Each lifetime receives fresh independent priorities in every replica,
and a stored snapshot retains these priorities.  Thus a deleted and
reinserted identity has distinct lifetimes at the two endpoints of a
snapshot comparison.  Every query contains at most one lifetime of
each identity and is evaluated as the corresponding set of identities.
Historical prefixes are queried through the global oracle defined above.

\subsubsection{Exact sampling at a fixed scale}
\label{sec:val-fixed-scale}

Fix distinct priorities, a scale $A>0$, and an intensity $c>0$.
Write $w_i(U)=w_i^\pi(U)$ and
$W^+(U)=\sum_i(w_i(U))_+$.  We maintain independent counts
\begin{equation}
 K_i\sim\operatorname{Poi}\!\left(\frac{c(w_i(U))_+}{A}\right).
 \label{eq:val-poisson-invariant}
\end{equation}
Only nonzero counts and their total $K=\sum_iK_i$ are stored.
Thus $K\sim\operatorname{Poi}(cW^+(U)/A)$.
The intensity is proportional to value, rather than to the number
of elements with positive weight.

Insertion decreases every old chain marginal, and deletion
increases every surviving marginal.  Consequently, insertion can
be handled by thinning the existing sample.  After a deletion, the
new positive mass must be sampled.  The following telescoping
identity makes this possible without examining every element.
For an interval $J$ in the priority order, let $X_{<J}$ and
$X_{\le J}$ be its preceding and ending prefixes.  Then
\begin{equation}
 B_X(J):=f(X_{\le J})-f(X_{<J})
       =\sum_{i\in X\cap J}w_i(X).
 \label{eq:val-interval}
\end{equation}

It is useful to state the update rule for two snapshots, as in
Figure~\ref{fig:val-transport}. This includes single insertions and
deletions and allows a sequence of local updates to be processed
in one batch.

\begin{lemma}[Snapshot transport]
\label{lem:val-transport}
Let $U_0,U_1$ be two snapshots, and let $C$ be their common
continuously surviving lifetimes.  Put
$D=U_0\setminus C$ and $I=U_1\setminus C$.
An exact sampler satisfying~\eqref{eq:val-poisson-invariant}
at $U_0$ can be transformed into an exact sampler at $U_1$ using
\begin{equation}
 O\!\left(
 \left[|D|+|I|+1+
 \frac{c(\OPT(U_0)+\OPT(U_1))}{A}\right]\log(N+1)
 \right)
 \label{eq:val-transport-cost}
\end{equation}
expected queries.
The snapshot choices must be fixed, or depend only on randomness
independent of the sampler's priorities and Poisson draws.
\end{lemma}
\begin{proof}
We perform the two artificial transitions
$U_0\longrightarrow C\longrightarrow U_1$.
An identity deleted and reinserted between the snapshots is absent
from $C$, even if it belongs to both endpoint sets of identities.

For the deletion transition, query the old weights $w_d(U_0)$ for
$d\in D$, and store their interval sums.  For $i\in C$, define
$d_i=w_i(C)-w_i(U_0)\ge0$.  Equation~\eqref{eq:val-interval} gives
\begin{equation}
 \sum_{i\in C\cap J}d_i
  =B_C(J)-B_{U_0}(J)+\sum_{d\in D\cap J}w_d(U_0).
 \label{eq:val-hole-correction}
\end{equation}
The last sum corrects for the deleted positions in the old chain.
The total drift satisfies
\[
 D_*:=\sum_{i\in C}d_i
 =f(C)-f(U_0)+\sum_{d\in D}w_d(U_0)
 \le\OPT(C)+\OPT(U_0),
\]
using nonnegativity and Lemma~\ref{lem:val-modular}.
After the hole weights have been queried, an interval drift
requires only constantly many value queries.  A balanced search
tree therefore samples $i$ with probability $d_i/D_*$ using
$O(\log(N+1))$ queries.  Intervals of zero drift are never selected.

Delete the points on $D$ and retain all old points on $C$.
Independently generate $\operatorname{Poi}(cD_*/A)$ proposals.
For each proposal, sample $i$ proportionally to $d_i$ and accept
the new point with probability
\[
 \frac{(w_i(C))_+-(w_i(U_0))_+}{d_i}.
\]
The ratio is in $[0,1]$, since the positive-part map is monotone
and $1$-Lipschitz.  Poisson marking and thinning give independent
new counts with precisely the required increments in their means.
Superposition yields the exact process on $C$.

For the insertion transition, retain each point at $i\in C$
independently with probability
\[
 \frac{(w_i(U_1))_+}{(w_i(C))_+}.
\]
A coordinate carrying a point has a positive denominator almost
surely.  Query the new weights at the elements of $I$ and give
them independent counts of means $c(w_i(U_1))_+/A$.
This establishes~\eqref{eq:val-poisson-invariant} on $U_1$.

All new draws and retention coins are independent.  The expected
number of deletion proposals is $cD_*/A$, and the expected number
of points inspected during insertion is $cW^+(C)/A$.
Since $\OPT(C)\le\min\{\OPT(U_0),\OPT(U_1)\}$, the hole queries,
birth queries, and sampled searches give~\eqref{eq:val-transport-cost}.
No conditioning on small sample sizes is used.
\end{proof}

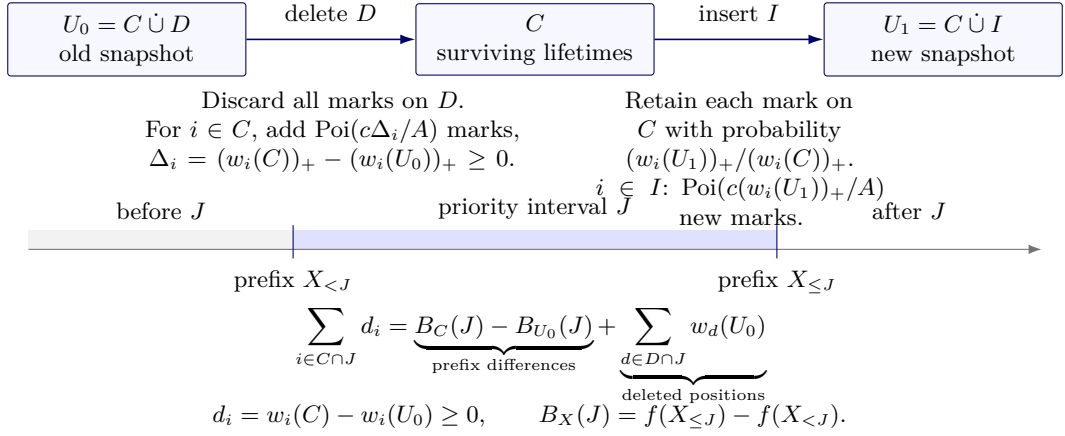
\begin{figure}[tbp]
\centering
\begin{tikzpicture}[
  font=\footnotesize,
  snapshot/.style={draw=blue!35!black,fill=blue!3,
    rounded corners=1.2pt,minimum width=3.15cm,
    minimum height=.84cm,align=center,inner sep=4pt},
  transport/.style={-{Latex[length=2mm,width=1.3mm]},
    draw=blue!40!black,semithick},
  operation/.style={align=center,text width=5.25cm,
    inner sep=0pt,font=\footnotesize}
]
\node[snapshot] (old) at (0,0)
  {$U_0=C\mathbin{\dot\cup}D$\\old snapshot};
\node[snapshot] (common) at (5.4,0)
  {$C$\\surviving lifetimes};
\node[snapshot] (new) at (10.8,0)
  {$U_1=C\mathbin{\dot\cup}I$\\new snapshot};
\draw[transport] (old.east) -- node[above=3pt] {delete $D$}
  (common.west);
\draw[transport] (common.east) -- node[above=3pt] {insert $I$}
  (new.west);
\node[operation,anchor=north] at (2.7,-.68) {
  Discard all marks on $D$.\\
  For $i\in C$, add $\operatorname{Poi}(c\Delta_i/A)$ marks,\\
  $\Delta_i=(w_i(C))_+-(w_i(U_0))_+\ge0$.
};
\node[operation,anchor=north] at (8.1,-.68) {
  Retain each mark on $C$ with probability\\
  $(w_i(U_1))_+/(w_i(C))_+$.\\
  $i\in I$: $\operatorname{Poi}(c(w_i(U_1))_+/A)$ new marks.
};
% The interval endpoints determine the two prefix evaluations.
\fill[black!5] (-1.3,-2.78) rectangle (2.2,-2.52);
\fill[blue!12] (2.2,-2.78) rectangle (8.6,-2.52);
\draw[draw=black!55,-{Latex[length=1.7mm]}]
  (-1.3,-2.78) -- (12.1,-2.78);
\draw[draw=blue!35!black] (2.2,-2.9) -- (2.2,-2.46);
\draw[draw=blue!35!black] (8.6,-2.9) -- (8.6,-2.46);
\node[anchor=south] at (.45,-2.51) {before $J$};
\node[anchor=south] at (5.4,-2.51) {priority interval $J$};
\node[anchor=south] at (10.35,-2.51) {after $J$};
\node[anchor=north] at (2.2,-2.93) {prefix $X_{<J}$};
\node[anchor=north] at (8.6,-2.93) {prefix $X_{\le J}$};
\node[align=center,anchor=north,inner sep=0pt] at (5.4,-3.52) {
  $\displaystyle
  \sum_{i\in C\cap J}d_i
  =\underbrace{B_C(J)-B_{U_0}(J)}_{\text{prefix differences}}
   +\underbrace{\sum_{d\in D\cap J}w_d(U_0)}_{\text{deleted positions}}$
};
\node[align=center,anchor=north,inner sep=0pt] at (5.4,-4.82) {
  $d_i=w_i(C)-w_i(U_0)\ge0,\qquad
   B_X(J)=f(X_{\le J})-f(X_{<J})$.
};
\end{tikzpicture}
\caption{Transport between two snapshots. Deletion raises surviving
marginals and is handled by Poisson immigration; insertion lowers them
and is handled by thinning. The interval identity supplies the drift
used to sample immigration proposals, accepted at $i$ with probability
$\Delta_i/d_i$. A deleted and reinserted identity contributes distinct
lifetimes to $D$ and $I$.}
\label{fig:val-transport}
\end{figure}

The lemma holds for every fixed order.  Conditional on an independent
update schedule and on the priorities, it therefore preserves the
exact product-Poisson law throughout a fully dynamic sequence.
At scales $A$ comparable to the optimum, its expected cost is already
$O((c+1)\log N)$ per update.  The remaining issue is to avoid paying
this cost multiplied by $\OPT/A$ at the smaller scales.

\subsubsection{Value windows and independent refresh times}
\label{sec:val-windows}

We give the construction with parameters, and choose them together
in Section~\ref{sec:val-parameters}.  Let
\[
 0<\tau<1/4,\qquad \lambda\ge1,\qquad \kappa\ge2\lambda.
\]
The constants $\lambda,\kappa$ will be absolute.
Cache the positive singleton gain $v_i$ on the first arrival of each
identity and ignore elements with $v_i=0$. By
Lemma~\ref{lem:prelim-singletons}, the remaining active universe satisfies
\begin{equation}
 \OPT(U)\le f_0+\sum_{i\in U}v_i.
 \label{eq:val-singleton-bound}
\end{equation}
An insertion or deletion of an element with $v_i\le A$ changes
the optimum of a restricted universe by at most $A$.

For each dyadic scale $A=2^j\ge f_0$, maintain
\begin{equation}
 W_A=\{i\in U:\tau A/N\le v_i\le A\},\qquad
 \ell_\tau=\left\lceil\log_2(N/\tau)\right\rceil+2.
 \label{eq:val-windows}
\end{equation}
An element belongs to at most $\ell_\tau$ windows.
Create an instance when its first eligible lifetime arrives, and
discard it when the window becomes empty.  An empty window has
the exact zero point process.  At a nonempty window we maintain
$q$ independent priority replicas and Poisson processes.
All replicas of a window use the same refresh schedule.

At the beginning of a period, first make all replicas current.
Draw a fresh independent half-subset $R$ of the window and set
\begin{equation}
 B=\max\{\lambda f(R),f_0,\max_{i\in W_A}v_i\},
 \qquad
 r=\max\{1,\lfloor B/(\kappa A)\rfloor\}.
 \label{eq:val-period}
\end{equation}
Record the next $r$ local updates without transporting the sampler
after each one.  At the $r$th update, use
Lemma~\ref{lem:val-transport} to move directly to the new snapshot,
and start another period with a fresh sample.
Updates outside the window do not advance its local clock.
All refreshes due at a global update are completed before reporting.

The half-sample is used only to choose a period length.  It need
not be an accurate estimate with high probability.  If $m_0$ is the
window optimum at the period start, nonnegative submodular sampling
gives
\begin{equation}
 \E[B\mid\text{history before the fresh sample}]
       \ge\frac{\lambda m_0}{4},
 \qquad B\le\lambda m_0.
 \label{eq:val-period-moments}
\end{equation}
The lower bound follows from Lemma~\ref{lem:prelim-sampling};
the upper bound is deterministic because every sampled set lies in
the current window and $\lambda\ge1$.

The next lemma is the reason that this coarse estimate suffices.
A stale instance must have large optimum relative to its scale.
Thus an instance at a suitable reporting scale is always current,
regardless of the sampling outcome used to choose its period.

\begin{lemma}[Freshness and period length]
\label{lem:val-period}
If $\OPT(W_A)\le A$, the sampler of $W_A$ is current.
Moreover, with $z=m_0/A$ at a period start,
\begin{equation}
 \E[r\mid\text{history before the fresh sample}]
       \ge\frac{\lambda}{16\kappa}(1+z),
 \qquad r\le1+\frac{\lambda}{\kappa}z.
 \label{eq:val-period-length}
\end{equation}
For a stream starting empty, the sum of the full lengths of all
periods started by any fixed time $T\ge1$ is $O(T\ell_\tau)$,
including the final unfinished periods.
\end{lemma}
\begin{proof}
A stale period has $r\ge2$, so $B\ge2\kappa A$.
Fewer than $r$ local updates have occurred, and each changes the
window optimum by at most $A$.  Its current optimum is therefore
at least
\[
 m_0-rA
 \ge(1-\lambda/\kappa)m_0
 \ge2(\kappa/\lambda-1)A
 \ge2A.
\]
This proves the freshness assertion.

Let $\chi=\lambda/(4\kappa)$, so $0<\chi\le1/8$.
By~\eqref{eq:val-period-moments},
\[
 \E r\ge\max\{1,\chi z-1\}
          \ge\frac{\chi}{4}(1+z).
\]
The last inequality follows by separating $z\le2/\chi$ and
$z>2/\chi$.  The deterministic upper bound in
\eqref{eq:val-period-length} is immediate.

Completed periods account for at most $T\ell_\tau$ local updates.
For an unfinished period with $r\ge2$, write
$m=\OPT(W_A)$ for its current optimum.  The preceding argument gives
$m\ge(1-\lambda/\kappa)m_0$, and hence
\[
 r\le\frac{\lambda}{\kappa-\lambda}\frac mA.
\]
A period of length one contributes at most one.
There are at most $|U|\ell_\tau$ nonempty windows.  Also, by
\eqref{eq:val-singleton-bound} and geometric summation,
\begin{align*}
 \sum_{A:W_A\ne\varnothing}\frac{\OPT(W_A)}A
 &\le\sum_{A\ge f_0}\frac{f_0}{A}
       +\sum_{i\in U}\sum_{A:i\in W_A}\frac{v_i}{A}\\
 &\le2+2|U|.
\end{align*}
The first sum contributes zero when $f_0=0$.
Thus the full lengths of the unfinished periods sum to
$O(|U|\ell_\tau+1)=O(T\ell_\tau)$.
An instance that becomes empty is reset exactly.  A period with
$r\ge2$ cannot become empty while unfinished, by freshness and
$f_0\le A$, so these resets introduce no uncharged periods.
\end{proof}

\begin{lemma}[Expected update complexity]
\label{lem:val-amortization}
For fixed $\lambda,\kappa$ as above, all windows and $q$ replicas
can be maintained through $T$ updates from an initial universe
of size $n_0$ using
\begin{equation}
 O\!\left(1+(n_0+T)q\ell_\tau(c+1)\log(N+1)\right)
 \label{eq:val-total-cost}
\end{equation}
expected queries.
\end{lemma}
\begin{proof}
First consider a stream starting empty.  At a period start let
$z=m_0/A$.  During its $r$ local updates at most $r$ lifetimes
change, and the next endpoint optimum is at most $A(z+r)$.
Conditional on the independent schedule, the expected transport
cost from Lemma~\ref{lem:val-transport} is
\[
 O\bigl(q(c+1)(1+z+r)\log(N+1)\bigr).
\]
The condition that a period has started by time $T$ is determined
before its fresh scheduling sample is drawn.  We may therefore
sum conditional expectations over all such starts and use
\eqref{eq:val-period-length}:
\[
 \E\sum_{\text{starts by }T}(1+z)
 \le\frac{16\kappa}{\lambda}
       \E\sum_{\text{starts by }T}r
 =O(T\ell_\tau).
\]
We have included the final unfinished periods even though their
transport work has not yet been performed.
The last assertion of Lemma~\ref{lem:val-period} bounds this terminal
charge; omitting it would not establish a bound for every prefix.

There are $O(T\ell_\tau)$ period starts and instance creations.
The one scheduling query per period, singleton queries, and the
initialization of each new one-element instance are absorbed in
\eqref{eq:val-total-cost}.  Process the initial universe as $n_0$
virtual insertions before time zero and apply the same argument
to $n_0+T$ updates.  The initial query for $f_0$ accounts for
the additive constant.
\end{proof}

The proof is unconditional.  An inaccurate period estimate may
increase the work of a particular refresh, but the expected period
length pays for that work.  Fresh independence at every period
start is essential for this argument.

\subsubsection{Choosing a current instance}
\label{sec:val-selection}

Let $a_*=\max\{f_0,\max_{i\in U}v_i\}$.  If $a_*=0$, the optimum
is zero by~\eqref{eq:val-singleton-bound}, and the algorithm reports
zero.  Otherwise start at the least dyadic $A_0\ge a_*$ and examine
$A_0,2A_0,4A_0,\ldots$ until the first current window is found.
Currentness is a stored flag, so no oracle queries are used.
An empty or never-created window counts as current.

\begin{lemma}[Selected restriction]
\label{lem:val-selection}
For $M=\OPT(U)>0$, the selector makes $O(\log(N+1))$ flag tests
and chooses an exact current instance at a scale $A$ such that
\begin{equation}
 a_*\le A<2M,\qquad
 \OPT(W_A)\ge(1-2\tau)M.
 \label{eq:val-selected-guarantee}
\end{equation}
Conditional on the separate scheduling randomness, the selected
scale and restriction are independent of every replica's priorities
and Poisson draws.
\end{lemma}
\begin{proof}
Let $A'$ be the least dyadic scale at least $M$.
Its window has optimum at most $A'$, so it is current by
Lemma~\ref{lem:val-period}.  The selector therefore stops at
$A\le A'<2M$.
Also $M\le f_0+\sum_i v_i\le(N+1)a_*$, which bounds the number
of scales examined.

Since all $v_i\le A$, the upper window threshold omits no element.
The sum of the gains below the lower threshold is at most
$\tau A<2\tau M$.  If $O$ is optimal, submodularity gives
\[
 f(O)-f(O\cap W_A)\le
 \sum_{i\in O\setminus W_A}v_i<2\tau M.
\]
This proves the approximation of the restriction.
Finally, the schedule and the selector depend only on the update
sequence, singleton values, and fresh independent half-samples.
They use no priority or Poisson randomness.
\end{proof}

The selected scale may be smaller than the optimum.  This is
harmless for reporting because the total point count is already
stored.  We never scan the selected sample merely to read its
total.  A confidence correction that depends on this count will
handle every possible ratio $M/A$.

\subsubsection{Reporting and the choice of parameters}
\label{sec:val-parameters}

We use the following elementary consequence of Poisson concentration.

\begin{lemma}[Poisson lower estimate]
\label{lem:val-confidence}
For $K\sim\operatorname{Poi}(\mu)$ and $x>0$, put
$R(K,x)=2\sqrt{Kx}+4x$.
Except with probability $2e^{-x}$,
\[
 |K-\mu|\le R(K,x),\qquad
 R(K,x)\le3\sqrt{\mu x}+7x.
\]
\end{lemma}
\begin{proof}
The Poisson Chernoff bound gives
$|K-\mu|\le\sqrt{2\mu x}+2x/3$ except with probability $2e^{-x}$.
Its lower-tail side implies
$\sqrt\mu\le\sqrt K+2\sqrt x$, and its upper-tail side implies
$K\le2\mu+2x$.  Substitution gives the two assertions.
\end{proof}

Let $W=W_A$ be the selected restriction.  Replica $j$ has a stored
count $K_j$ and a true certificate
$Q_j=f_0+W_j^+(W)$.  Report
\begin{equation}
 L_t=\max_{1\le j\le q}
 \left\{f_0+\frac Ac[K_j-R(K_j,x)]_+\right\}.
 \label{eq:val-report}
\end{equation}
On its confidence event, replica $j$'s contribution
$\underline Q_j$ satisfies
\begin{align*}
 0\le Q_j-\underline Q_j
 &\le\frac{2A}{c}R(K_j,x)\\
 &\le6\sqrt{\frac{A W_j^+(W)x}{c}}+\frac{14Ax}{c}\\
 &\le\omega(c,x)M,\qquad
 \omega(c,x):=6\sqrt{\frac{2x}{c}}+\frac{28x}{c}.
\end{align*}
Here we used only $A<2M$ and $W_j^+(W)\le M$.
Thus the correction does not require $A$ to upper-bound the optimum.

The remaining loss comes from obtaining a sufficiently good random
order.  Let $0<\zeta<\alpha$.  Since
$Q_j\in[0,\OPT(W)]$ and $\E Q_j\ge\alpha\OPT(W)$,
\[
 \Pr\{Q_j\ge(\alpha-\zeta)\OPT(W)\}
 \ge\frac{\zeta}{1-\alpha+\zeta}\ge\zeta.
\]
Conditional on the schedule, the $q$ orders are independent and
uniform on the selected restriction.  Their maximum therefore
meets this bound except with probability $e^{-q\zeta}$.
Conditional also on the priorities, Lemma~\ref{lem:val-confidence}
applies to the exact Poisson laws.  A union bound over the replicas
gives the following parameterized statement.

\begin{proposition}[Value estimator with parameters]
\label{prop:val-parameters}
At each fixed update time, with probability at least
$1-e^{-q\zeta}-2qe^{-x}$, the report satisfies
\begin{equation}
 \bigl(\alpha-\zeta-2\alpha\tau-\omega(c,x)\bigr)\OPT(U_t)
 \le L_t\le\OPT(U_t).
 \label{eq:val-parametric-accuracy}
\end{equation}
Its total expected query bound is~\eqref{eq:val-total-cost}.
\end{proposition}
\begin{proof}
On the intersection of the priority and confidence events,
\[
 \OPT(U_t)\ge L_t
 \ge(\alpha-\zeta)\OPT(W)-\omega(c,x)\OPT(U_t).
 \]
Apply~\eqref{eq:val-selected-guarantee} and drop the nonnegative
term $2\tau\zeta\OPT(U_t)$.
The probability bounds above hold conditionally on the independent
schedule, so they also hold without conditioning.
Lemma~\ref{lem:val-amortization} supplies the query bound.
\end{proof}

We now choose the parameters once.  For
$0<\varepsilon<1/10$ and $0<\delta<1/2$, take
\begin{equation}
 \begin{gathered}
 \lambda=4,\qquad \kappa=4\lambda,\qquad
 \tau=\varepsilon/8,\qquad \zeta=\varepsilon/4,\\
 q=\left\lceil\zeta^{-1}\log(2/\delta)\right\rceil,\qquad
 x=\log(4q/\delta),\qquad
 c=(32/\varepsilon)^2x.
 \end{gathered}
 \label{eq:val-parameter-choice}
\end{equation}
These choices satisfy $\kappa\ge2\lambda$ and all the stated range
conditions.  They also give
\[
 e^{-q\zeta}\le\delta/2,\qquad 2qe^{-x}=\delta/2,
 \qquad
 \zeta+2\alpha\tau<\varepsilon/2.
\]
Finally,
\[
 \omega(c,x)
 =\frac{3\sqrt2}{16}\varepsilon
       +\frac7{256}\varepsilon^2
 <\varepsilon/2.
\]
Substitution into Proposition~\ref{prop:val-parameters} proves
the value-estimation theorem below.

\begin{theorem}[Fully dynamic value estimation]
\label{thm:val-dynamic}
For every $0<\varepsilon<1/10$ and $0<\delta<1/2$, there is a
fully dynamic algorithm such that, at each fixed time $t$,
\[
 \Pr\!\left\{
   (\sqrt2-1-\varepsilon)\OPT(U_t)\le L_t\le\OPT(U_t)
       \right\}\ge1-\delta.
\]
It uses
\[
 O\!\left(
 \varepsilon^{-3}\log(1/\delta)\log(q/\delta)
 \log(N/\varepsilon)\log(N+1)\right)
 \]
expected amortized oracle queries per update, with the same bound
per initial element in preprocessing.  For fixed
$\varepsilon,\delta$, this is $O(\log^2(N+1))$.
\end{theorem}

The guarantee has no dependence on the stream length or on the
range of nonzero oracle values.  It holds for every fixed prefix
of an arbitrary oblivious sequence of insertions, deletions, and
reinsertions.  If simultaneous accuracy through a prescribed
horizon $H$ is required, replace $\delta$ by $\delta/(H+1)$ and
take a union bound over the reporting times.  This introduces
only logarithmic factors in $H$ into the query bound.

\subsection{Limitation of the uniform-chain certificate}
\label{sec:val-sharpness}

The factor in Theorem~\ref{thm:val-chain} is tight even on directed
cut functions.  The example also separates the value of the
certificate from the value of its positive support.  Thus an
improvement cannot come solely from a sharper analysis of the same
statistic, although a different statistic or a different ordering
may avoid the example.

For $b\in[0,1]$, take disjoint groups $X,Y$, each of size $m$,
and define
\begin{equation}
 f_m(S)=x(1-y)+b\,y(1-x),\qquad
 x=\frac{|S\cap X|}{m},\qquad y=\frac{|S\cap Y|}{m}.
 \label{eq:val-tight-family}
\end{equation}
This is the outgoing cut function of a bidirected complete bipartite
graph: each edge from $X$ to $Y$ has weight $1/m^2$, and each
reverse edge has weight $b/m^2$.
It is nonnegative and submodular, and
$f_m(\varnothing)=f_m(X\cup Y)=0$.
The displayed function is affine in each of $x,y$, so its maximum
on the unit square is attained at a corner.  Hence its optimum
is $f_m(X)=1$.

In a uniform order, a prefix at priority $p$ contains approximately
a $p$ fraction of each group.  Put $s=1+b$.
The corresponding marginals, multiplied by $m$, approach
$1-sp$ on $X$ and $b-sp$ on $Y$.  The positive-chain value is
therefore governed by the two areas
\begin{equation}
 g(b):=\int_0^1(1-sp)_+\,dp+\int_0^1(b-sp)_+\,dp
      =\frac{1+b^2}{2(1+b)}.
 \label{eq:val-tight-limit}
\end{equation}
We make both the limit and the effect of replication explicit.

\begin{proposition}[Tight family]
\label{prop:val-sharpness}
Let $Q_m$ be the uniform-chain certificate of
\eqref{eq:val-tight-family}, and let $s=1+b$.  Then
\[
 g(b)\le\E Q_m\le g(b)+\frac{\pi s}{4\sqrt m},
 \qquad
 \Pr\{Q_m>\E Q_m+t\}\le
       \exp(-mt^2/s^2).
 \]
Consequently, for any $K$ uniform-priority replicas and
$\delta\in(0,1)$, with probability at least $1-\delta$,
\begin{equation}
 \max_{1\le r\le K}Q_m^{(r)}
 \le g(b)+\frac{\pi s}{4\sqrt m}
              +s\sqrt{\frac{\log(K/\delta)}m}.
 \label{eq:val-portfolio-bound}
\end{equation}
Independence between the replicas is unnecessary; each replica
must have independent uniform priorities within its own order.
\end{proposition}
\begin{proof}
Condition on the priority $p$ of an element in $X$.
The number $Z$ of preceding $Y$ elements is
$\operatorname{Bin}(m,p)$, and its marginal is
$(1-sZ/m)/m$.  Similarly, a $Y$ marginal is
$(b-sZ'/m)/m$, where $Z'$ has the same binomial distribution.
Thus
\[
 \E Q_m=
 \int_0^1\E(1-sZ/m)_+\,dp+
 \int_0^1\E(b-sZ/m)_+\,dp.
 \]
For either $a\in\{1,b\}$, convexity and the Lipschitz property of
the positive part imply
\[
 0\le\E(a-sZ/m)_+-(a-sp)_+
 \le s\E|Z/m-p|
 \le s\sqrt{p(1-p)/m}.
 \]
Integrating and using $\int_0^1\sqrt{p(1-p)}\,dp=\pi/8$
proves the expectation bound.

Changing one priority changes the certificate by at most $s/m$.
Indeed, suppose an $X$ element moves later and crosses $k$
elements of $Y$.  Its own positive contribution decreases by
an amount in $[0,sk/m^2]$.  The sum of the crossed elements'
positive contributions increases by an amount in the same
interval.  All other terms are unchanged, so the absolute change
in the total is at most $sk/m^2\le s/m$.
Moving earlier reverses the comparison, and the argument for a
$Y$ element is identical.
The bounded-differences inequality~\cite{McDiarmid89}, applied to
the $2m$ independent priorities, gives
\[
 \Pr\{Q_m-\E Q_m>t\}
 \le\exp\!\left(-\frac{2t^2}{2m(s/m)^2}\right).
 \]
Apply this bound to each replica with
$t=s\sqrt{\log(K/\delta)/m}$ and take a union bound.
\end{proof}

\paragraph{Choice of the edge-weight parameter.}
Differentiating~\eqref{eq:val-tight-limit} shows that its minimum
on $[0,1]$ occurs at $b=\sqrt2-1$.  At this choice,
$b^2+2b=1$, so $g(b)=b=\alpha$.
Equation~\eqref{eq:val-portfolio-bound} implies that the best of
polynomially many uniform replicas is at most
$\alpha+o(1)$ with high probability.
Thus even replication cannot improve the approximation ratio of
the uniform raw-chain certificate by a fixed constant.

The value windows do not discard either group on this example.
Its singleton gains are $1/m$ and $b/m$.  For the choice
in~\eqref{eq:val-parameter-choice}, $\tau<b$.
Any selected scale satisfies $1/m\le A<2$; since $N\ge2m$,
its lower threshold is
$\tau A/N<\tau/m<b/m$.
Consequently its retained window is the entire universe.

This conclusion concerns the certificate, rather than the optimum
value attainable by other algorithms.  For example, placing all
of $X$ before all of $Y$ gives certificate one.
Moreover, under a uniform order, the positive support has limiting
fractions $x=1/(1+b)$ and $y=b/(1+b)$, by uniform convergence of
the empirical priority distributions.  Its limiting value is
\[
 \frac{1+b^3}{(1+b)^2},
\]
which is strictly larger than $\alpha$ at the minimizing choice
of $b$.  The certificate loses value because an element's raw
marginal is evaluated in the presence of preceding elements that
will ultimately be omitted.

The dynamic sampler does not construct this entire support.
For instance, on the modular function $f(S)=|S|/N$, the full
certificate equals one, while a Poisson sample of bounded intensity
at scale one reveals only a bounded expected number of identities.
Its marked set has expected value at most $c/N$.
Estimating the positive mass therefore does not, by itself, maintain
a feasible set with comparable value.  The next part develops
algorithms that maintain such a set.

\section{Maintaining solutions with square-root updates}
\label{set-sec-root-algorithms}
\label{sec:root-updates}

We now require the data structure to maintain a feasible set.
All queries in this section involve only currently active elements.
The algorithm has a decremental core and a small insertion buffer.
We maintain a solution on the core and reoptimize its extension to
the buffer after every update. Choosing a square-root rebuilding
period balances the cost of maintaining the core against the cost
of optimizing the buffer.

Two ingredients come first. Weighted double greedy combines the
core with its insertion buffer. A decremental mask structure
maintains the core's chain solution; we will also use this structure
in the general tradeoff. We keep the parameters symbolic until
Subsection~\ref{set-sec-root-parameters}.

\subsection{Double greedy with a forced base}
We require an endpoint-sensitive form of randomized double
greedy. The algorithm is the usual one~\cite{BFNS12,BFNS15};
we include the weighted inequality to make its later use explicit.

Fix disjoint sets $L,D$. Initialize $X=L$ and $Y=L\cup D$,
and process the elements of $D$ in any fixed order.
For the current element $i$, let
\[
 g=f(X\cup\{i\})-f(X),\qquad
 h=f(Y\setminus\{i\})-f(Y).
\]
Add $i$ to $X$ with probability $g_+/(g_++h_+)$ and otherwise
remove it from $Y$, using probability $1/2$ when both positive
parts vanish. The final common set $T=X=Y$ contains $L$ and
is contained in $L\cup D$. The procedure uses $O(|D|+1)$
value queries.

\begin{lemma}[Weighted double greedy]
\label{lem:prelim-rdg}
For every $L\subseteq Q\subseteq L\cup D$ and every
$\theta,\phi\ge0$ satisfying $\theta\phi\ge1/4$,
\[
 (1+\theta+\phi)\E f(T)
 \ge f(Q)+\theta f(L)+\phi f(L\cup D).
\]
The inequality remains valid conditionally when $L,D,Q$ are
fixed before the fresh randomness of double greedy is drawn.
\end{lemma}
\begin{proof}
Couple the comparator with the endpoints by setting
$Q_i=(Q\cup X_i)\cap Y_i$ after step $i$.
Before an element is processed, $X\subseteq Q_i\subseteq Y$.
Submodularity gives $g+h\ge0$.
If $g,h\ge0$ and $g+h>0$, the expected decrease in the
comparator's value is at most $gh/(g+h)$.
Indeed, if the element belongs to the comparator, only its
removal can cause a decrease, bounded by $g$; otherwise,
only its addition can cause a decrease, bounded by $h$.
The expected increases of the lower and upper endpoints
are respectively $g^2/(g+h)$ and $h^2/(g+h)$.
Since
\[
 gh\le\theta g^2+\phi h^2
 \qquad(\theta\phi\ge1/4),
\]
the comparator's expected loss is at most the corresponding
weighted sum of endpoint gains.

If one raw gain is negative, the algorithm chooses the other
endpoint deterministically and the comparator's loss is
nonpositive. If both gains vanish, the same inequality is
immediate. Summing over the steps and using
$Q_{\rm final}=X_{\rm final}=Y_{\rm final}=T$ gives
\[
 f(Q)-\E f(T)
 \le \theta\bigl(\E f(T)-f(L)\bigr)
     +\phi\bigl(\E f(T)-f(L\cup D)\bigr).
\]
Rearranging proves the claim.
\end{proof}

\subsection{Maintaining positive chain masks under deletions}
\label{set-sec-decremental}

The value estimator maintains only sampled mass. To maintain a feasible
solution, we must also identify enough of the positive coordinates.
This can be done efficiently when the active universe only shrinks.
Under deletions, every surviving chain marginal increases. Once an
element is known to have positive marginal, it remains safe to retain
it. The remaining task is to find positive coordinates that have not
yet been discovered.

We use the following general form of a chain vector. At the beginning
of an epoch, choose an order on the retained elements. Let $z_i$ be
either its raw forward marginal or the average of its forward and
reverse marginals. Both choices have three properties:
\begin{enumerate}
\item $z$ is a modular lower vector, with
      $\sum_i z_i=f(W)-f(\varnothing)$;
\item deleting elements can only increase surviving coordinates;
\item the sum of $z_i$ over an interval of the order can be obtained
      with a constant number of prefix or suffix queries.
\end{enumerate}
The averaging in the second choice does not change these properties.
We may also designate an eligibility set from the sampled priorities
at the beginning of the epoch. Only eligible positive coordinates
are required in the maintained mask.

The data structure stores a lower anchor $b_i\le z_i$ for each
coordinate. Its mask consists of the eligible coordinates with
$b_i>0$. An anchor is updated only when it becomes positive or when
its negative magnitude drops by a constant factor. This rule is what
bounds the total number of successful repairs. Random sampling detects
whether a significant amount of positive mass is still missing.

\begin{lemma}[Decremental mask maintenance]
\label{set-lem-mask-maintenance}
Fix $0<\xi<1/20$. A decremental instance initialized on $m$ active
identities maintains a retained universe $W\subseteq C$ and a stored
positive mask for each of a constant number of chain vectors of the
forms above. At each fixed time of an oblivious deletion sequence,
\begin{equation}\label{set-eq-mask-properties}
 \OPT(W)\ge(1-2\xi)\OPT(C),\qquad
 \E E\le2\xi\OPT(W),
\end{equation}
where $E$ is the eligible positive mass omitted by one mask. No
stored-mask coordinate has negative current marginal. The omitted
positive singleton gains sum to at most $2\xi\OPT(C)$.

Initialization and any prefix of the subsequent deletions use at most
\begin{equation}\label{set-eq-mask-cost}
 m\,q(N,\xi),\qquad
 q(N,\xi)=O\!\left(\xi^{-1}
                  \log^4\frac{N+1}{\xi}\right)
\end{equation}
oracle queries, for every outcome of the random choices. All query
sets are subsets of the current active cohort. The bound is independent
of the numerical range of positive function values.
\end{lemma}

\subsubsection{Value epochs and stored anchors}

Put $f_0=f(\varnothing)$ and cache the singleton gains
$s_i=f(\{i\})-f_0$ when the instance is initialized. Elements with
$s_i\le0$ may be excluded from the answer: every later marginal of
such an element is nonpositive.

At the beginning of a value epoch, let
\[
 L=\max\bigl(\{f_0\}\cup\{f(\{i\}):i\in C\}\bigr),\qquad
 \theta=\xi L/N,\qquad W=\{i\in C:s_i\ge\theta\}.
\]
If $L=0$, return the empty set. Otherwise initialize the required
chain vectors on $W$, set $b_i=z_i$, and store the eligible coordinates
with positive anchors. Initializing a constant number of vectors
uses $O(|W|)$ queries.

The epoch ends when the current largest singleton value, including
$f_0$, drops below $L/2$. The next epoch begins directly at the new
maximum. Thus no intermediate numerical scales need to be visited.
Epoch boundaries depend only on the fixed update sequence and cached
singleton values, and are independent of the sampled orders.

On a deletion, remove the identity from all membership lists and
anchor trees. If the epoch ends, rebuild on the surviving cohort.
Otherwise, a deletion outside $W$ requires no sign repair. A deletion
inside $W$ invokes the procedure below.

\subsubsection{Repairing a mask}

During an epoch the drift $d_i=z_i-b_i$ is nonnegative. An interval
sum of the current vector uses at most four value queries, while the
anchor sum is stored. Hence an interval drift sum is available with
constantly many queries. Descending a balanced priority tree samples
a coordinate with probability
\[
 \frac{d_i}{\sum_jd_j}
\]
using $O(\log(N+1))$ queries. Querying an individual current coordinate
also uses constantly many prefix and suffix queries.

A sampled coordinate is called \emph{productive} if
\[
 b_i\le0<z_i,
 \qquad\text{or}\qquad
 b_i\le-\theta\ \text{ and }\ z_i>b_i/2.
\]
On a productive draw, replace $b_i$ by $z_i$, update the anchor sums,
and add the coordinate to its mask if it is eligible and positive.
Reset a consecutive-failure counter to zero. An unproductive draw
increments the counter. Stop when the counter reaches a trial
parameter $\tau$, or immediately if the total drift vanishes.
The choice of $\tau$ is given at the end of this section.

\subsubsection{Accuracy}

The omitted positive singleton gains have total at most $N\theta=\xi L$.
Before the epoch ends, $L/2\le\OPT(C)$. Therefore
\begin{equation}\label{set-eq-truncation}
 \OPT(W)\ge(1-2\xi)\OPT(C),\qquad
 N\theta\le\frac{2\xi}{1-2\xi}\OPT(W)\le4\xi\OPT(W).
\end{equation}
These inequalities remain valid when $W$ is empty, since
$\OPT(\varnothing)=f_0$.

Write $M_W=\OPT(W)$, $P=\sum_i(z_i)_+$, and
$Q=\sum_i(-z_i)_+$. The modular lower-vector property gives
$P,Q\le M_W$. Every eligible positive coordinate missing from the
mask is productive and has drift at least its positive weight.
If its total missing mass is $E$, the unproductive drift is at most
\begin{align*}
 &\sum_{b_i>0}(z_i-b_i)
 +\sum_{\substack{b_i\le-\theta\\z_i\le b_i/2}}(z_i-b_i)
 +\sum_{\substack{-\theta<b_i<0\\z_i\le0}}(z_i-b_i)\\
 &\hspace{35mm}\le P-E+Q+N\theta.
\end{align*}
Since productive drift is at least $E$, the probability of a productive
draw is at least $E/(P+Q+N\theta)$. In particular, whenever
$E>\xi M_W$, this probability is at least
\[
 \frac{\xi}{2+4\xi}>\frac{\xi}{3}.
\]

Initially, $|b_i|\le\OPT(W_{\mathrm{initial}})\le(N+1)L$.
Each productive negative refresh halves the anchor's magnitude until
it exceeds $-\theta$, after which at most one positive crossing remains.
Thus the number of productive rounds, together with their final
sampling phase, in any one epoch is at most
\[
 H(N,\xi):=
 N\left(\left\lceil\log_2\frac{2(N+1)^2}{\xi}\right\rceil+3\right)+1.
\]
Suppose that the integer $\tau$ satisfies
\begin{equation}\label{set-eq-trial-condition}
 H(N,\xi)\exp(-\xi\tau/3)\le\xi.
\end{equation}
A union bound over the possible productive rounds shows that a repair
stops with $E>\xi M_W$ with probability at most $\xi$. The argument
conditions on all earlier history and uses fresh drift samples. It
does not require earlier repairs to have been successful. Because
$0\le E\le M_W$ always, we obtain $\E E\le2\xi M_W$.

\subsubsection{Query complexity and the trial parameter}

An identity $i$ is retained only during epochs whose initial scales
satisfy
\[
 f(\{i\})\le L\le Ns_i/\xi.
\]
Successive epoch scales decrease by more than a factor of two.
Consequently each initial identity belongs to
$O(\log((N+1)/\xi))$ retained epochs, irrespective of numerical value
ratios. Initialization over all epochs costs
$O(m\log((N+1)/\xi))$ queries. Negative-anchor halvings contribute
$O(m\log^2((N+1)/\xi))$ productive refreshes.

Each productive refresh and each deletion pays for at most one block
of $\tau$ samples, and each sample uses $O(\log(N+1))$ queries.
We choose
\[
 \tau=\left\lceil\frac3{\xi}\log\frac{H(N,\xi)}{\xi}\right\rceil.
\]
This satisfies \eqref{set-eq-trial-condition}, and is
$O(\xi^{-1}\log((N+1)/\xi))$. Combining the bounds proves
Lemma~\ref{set-lem-mask-maintenance}. The accounting charges only
work already performed, so it also covers a module destroyed during
its final epoch. Membership lists are updated before a repair returns;
enumerating the stored mask makes no oracle calls.

\subsection{Square-root updates via a core and an insertion buffer}
\label{set-sec-sqrt}

We first give the stronger approximation available with square-root
update queries. Partition the active universe into an old core and a
small buffer of new elements. During one epoch the core receives only
deletions, so Lemma~\ref{set-lem-mask-maintenance} maintains a useful
core solution. The buffer is small enough to optimize anew after each
update.

Simply keeping the core solution and optimizing the buffer separately
does not suffice for a nonmonotone objective: taking their union can
destroy value. We instead extend two different core sets by randomized
double greedy on the buffer. One is the maintained core solution;
the other is an independent random sample of the core. The first
preserves the information obtained by optimization, whereas the
second provides the sampling inequalities needed to control
interactions across the two parts.

\subsubsection{A decremental core with endpoint terms}

Let $\alpha=\sqrt2-1$ and $\beta=(1-\alpha)/2$. The endpoint terms in
Theorem~\ref{thm:val-chain} are useful here: a lower bound that keeps
$f(C)$ can be combined with the loss from forcing a core set into a
buffer extension.

\begin{theorem}[Decremental core profile]
\label{set-thm-core-profile}
For every $0<\varepsilon<1/5$, a decremental instance initialized on
$m$ identities maintains a stored set $S\subseteq C$ such that
\begin{equation}\label{set-eq-core-profile}
 \E f(S)\ge
 \alpha\OPT(C)+\beta[f(\varnothing)+f(C)]
       -\varepsilon\OPT(C)
\end{equation}
at every fixed time of an oblivious deletion sequence. Initialization
and any deletion prefix together use
$m\,O(\varepsilon^{-1}\log^4((N+1)/\varepsilon))$ queries,
deterministically. All queries concern current active subsets.
\end{theorem}

\begin{proof}
Apply Lemma~\ref{set-lem-mask-maintenance} to one raw forward chain,
with all coordinates eligible. At each value-epoch start, draw independent
uniform priorities for the retained identities and use their increasing
order. The epoch boundaries and current retained universe are independent
of these priorities, so their restriction at every fixed time remains a
uniform random order. Its maintained mask contains only
positive current coordinates, so its value is at least the full
chain certificate minus the missing positive mass. On the current
retained universe $W$, Theorem~\ref{thm:val-chain} gives
\[
 \E f(S)\ge
 \alpha\OPT(W)+\beta[f_0+f(W)]-2\xi\OPT(W).
\]
Let $e$ be the sum of positive singleton gains omitted from $W$.
Diminishing returns gives
\[
 \OPT(C)\le\OPT(W)+e,\qquad f(C)\le f(W)+e,
 \qquad e\le2\xi\OPT(C).
\]
Removing elements of nonpositive singleton gain cannot hurt either
inequality. It follows that
\begin{align*}
 \E f(S)
 &\ge\alpha\OPT(C)+\beta[f_0+f(C)]
       -(\alpha+\beta)e-2\xi\OPT(W)\\
 &\ge\alpha\OPT(C)+\beta[f_0+f(C)]
       -(2+\sqrt2)\xi\OPT(C).
\end{align*}
Choose $\xi$ to satisfy $(2+\sqrt2)\xi\le\varepsilon$; for example,
$\xi=\varepsilon/4$ suffices. The cost follows from
\eqref{set-eq-mask-cost}.
\end{proof}

\subsubsection{Two extensions of the core}

The only randomized double-greedy property we need is the weighted
endpoint inequality in Lemma~\ref{lem:prelim-rdg}. Its weights are
analysis parameters; both extension calls use the ordinary randomized
double-greedy algorithm.

Fix $p\in(0,1)$ and positive weights $r,s,u,v$ satisfying
$rs=uv=1/4$. Put $q=1-p$ and define
\begin{align}
 d_R&=1+r+s,& d_S&=1+u+v,\notag\\
 \lambda&=sq-p,&
 L&=(1+s)q^2-rpq,& A&=(1+s)q.
 \label{set-eq-augmentation-parameters}
\end{align}
The conditions used in the analysis are
\begin{equation}\label{set-eq-augmentation-conditions}
 \lambda,L\ge0,\qquad
 \frac{L}{\alpha}\ge\lambda d_S,\qquad
 \frac{\beta}{\alpha}L+rp^2\ge\lambda(1+v).
\end{equation}
All parameter values are fixed in Subsection~\ref{set-sec-root-parameters}.

\begin{lemma}[Augmenting a maintained core]
\label{set-lem-augmentation}
Let $C,D$ partition a fixed universe $U$, let $M=\OPT(U)$, and suppose
that a stored random set $S\subseteq C$ satisfies
\[
 \E f(S)\ge
 \alpha\OPT(C)+\beta[f_0+f(C)]-\zeta\OPT(C).
\]
Independently sample $R\subseteq C$ with inclusion probability $p$.
Apply randomized double greedy on $[R,R\cup D]$ and on $[S,S\cup D]$,
and denote the outputs by $T_R,T_S$. Under
\eqref{set-eq-augmentation-conditions},
\begin{equation}\label{set-eq-augmentation-bound}
 \E\max\{f(S),f(T_R),f(T_S)\}
 \ge(\rho-\zeta)M,\qquad
 \rho=\frac{\alpha A}{\alpha d_R+L}.
\end{equation}
The two extensions and their value comparison use $O(|D|+1)$ queries.
\end{lemma}
\begin{proof}
Fix an optimum $O$ and put
\[
 X=O\cap C,\qquad Z=C\setminus O,\qquad
 Y=O\cap D,\qquad W=D\setminus O.
\]
These four sets may have arbitrary sizes.  Concatenated letters below
mean unions; thus $XY=O$, $XZ=C$, and $YW=D$.
For sampling probability $p$, write $q=1-p$ and
\[
 B_p(H)=q^2f(H)+pq f(XH)+pq f(ZH)+p^2f(CH).
\]
Concavity along the two disjoint sampling directions gives
$\mathbb E f(R\cup H)\ge B_p(H)$.

The weighted endpoint inequality for randomized double greedy gives
\begin{align}
 d_R\mathbb E f(T_R)
 &\ge B_p(Y)+rB_p(\varnothing)+sB_p(D),
 \label{set-eq-sqrt-random-extension}\\
 d_S\mathbb E f(T_S)
 &\ge\mathbb E f(SY)+u\mathbb E f(S)+v\mathbb E f(SD).
 \label{set-eq-sqrt-cached-extension}
\end{align}
For each realized $S\subseteq C$, submodularity supplies
\[
 f(SY)\ge f(CY)+f(S)-f(C),\qquad
 f(SD)\ge f(U)+f(S)-f(C).
\]
Consequently
\[
 d_S\mathbb E f(T_S)
 \ge f(CY)+d_S\mathbb E f(S)-(1+v)f(C)+v f(U).
\]

Add $\lambda$ times the preceding
cached-extension bound to \eqref{set-eq-sqrt-random-extension}.
Expand the three $B_p$ terms, and apply the following three
submodularity inequalities with the indicated nonnegative weights:
\begin{align*}
 pq\,[f(ZY)+f(D)]
 &\ge pq\,[f(Y)+f(ZD)],\\
 q\lambda\,[f(CY)+f(D)]
 &\ge q\lambda\,[f(Y)+f(U)],\\
 spq\,[f(CY)+f(XD)]
 &\ge spq\,[f(O)+f(U)].
\end{align*}
The available coefficients match exactly because
$pq+q\lambda=sq^2$ and
$q\lambda+spq=p^2+\lambda$.
Finally use $f(Y)\ge M-f(X)+f_0$, another submodularity inequality.
The resulting bound, after dropping only nonnegative terms, is
\begin{align}
 d_R\mathbb E f(T_R)+\lambda d_S\mathbb E f(T_S)
 \ge{}& AM-Lf(X)+\lambda d_S\mathbb E f(S)
 \notag\\
 &+[rp^2-\lambda(1+v)]f(C)+d_Rq^2f_0.
 \label{set-eq-sqrt-augmentation-master}
\end{align}
For clarity, the terms discarded in this step are exactly
\[
 rpq f(Z)+(1+s)pq f(ZD)
       +[sp+\lambda(q+v)]f(U),
\]
and all their coefficients are nonnegative.

Since $\operatorname{OPT}(C)\le M$ and $\operatorname{OPT}(C)\ge f(X)$,
the core guarantee implies
$\mathbb E f(S)\ge \alpha f(X)+\beta[f_0+f(C)]-\zeta M$.
Substitution into
\eqref{set-eq-sqrt-augmentation-master} gives
\begin{align*}
 &\left(\frac L\alpha-\lambda d_S\right)\mathbb E f(S)
     +d_R\mathbb E f(T_R)+\lambda d_S\mathbb E f(T_S)\\
 &\quad\ge
 \left(A-\frac L\alpha\zeta\right)M
 +[\frac{\beta}{\alpha} L+rp^2-\lambda(1+v)]f(C)
 +[\frac{\beta}{\alpha} L+d_Rq^2]f_0.
\end{align*}
By \eqref{set-eq-augmentation-conditions}, all three output coefficients and both remaining endpoint coefficients
are nonnegative.  Divide by the sum $d_R+L/\alpha$ of the output
coefficients.  Since a maximum is at least every convex combination,
the approximation factor is at least
\[
 \frac{\alpha A}{\alpha d_R+L}-\zeta\frac{L}{\alpha d_R+L}
 \ge\frac{\alpha A}{\alpha d_R+L}-\zeta.
\]

\end{proof}
\subsubsection{Fully dynamic implementation}
\label{set-sec-sqrt-dynamic-schedule}

We now combine the decremental core profile with the augmentation
lemma.  The schedule uses actual update counts and requires no
information about future deletions.

\begin{theorem}[Square-root updates with a maintained set]
\label{set-thm-sqrt-dynamic}
Let $\rho_{\rm sq}$ be the constant obtained from \eqref{set-eq-augmentation-bound} by the choices in Subsection~\ref{set-sec-root-parameters}.
For every fixed $0<\varepsilon<1/5$, there is a fully dynamic algorithm
maintaining a stored set $S_t\subseteq U_t$ with
\[
 \mathbb E f(S_t)\ge(\rho_{\rm sq}-\varepsilon)\operatorname{OPT}(U_t)
\]
at every fixed time of an arbitrary oblivious update sequence.
Its deterministic total query cost over any prefix of $T$ updates is
\begin{equation}\label{set-eq-sqrt-dynamic-total}
 O\!\left(1+n_0q(N,\varepsilon)
       +T\left[\frac Nd q(N,\varepsilon)+d\right]\right),
 \quad
 d=\lceil\sqrt N\rceil,
 \quad
 q(N,\varepsilon)=
 O\!\left(\varepsilon^{-1}
       \log^4\frac{N+1}{\varepsilon}\right).
\end{equation}
In particular, the amortized query bound is
$\widetilde O_\varepsilon(\sqrt N)$.
All queries use current active subsets.  The stored output can be
enumerated without further oracle calls.
\end{theorem}

\begin{proof}
Partition updates into consecutive epochs of $d$ updates.  At time
zero, and immediately after each multiple of $d$ updates, place
every currently active lifetime in a core $C$, initialize the
decremental data structure of
Theorem~\ref{set-thm-core-profile}, and empty an insertion
buffer $D$.  During the epoch, every insertion enters $D$.
A deletion is sent to the part containing its current lifetime.
Reinsertion starts a new lifetime in $D$; it never inserts into an
existing core module.  Thus the current $C,D$ always partition
$U_t$, the core sees only deletions, and $|D|\le d$.

Choose the module's error parameter to be $\varepsilon$; the
augmentation lemma multiplies this error by a coefficient below one.
Its lifetime bound therefore
has the form $m q(N,\varepsilon)$.
At the epoch's start, draw independent Bernoulli-$p$ membership
bits for the core, independently of the core module, and call the
resulting sample $R$.  Maintain $R$ by removing deleted core
lifetimes.  At each fixed time it has the required product
distribution on the current core and is independent of
the cached core solution: the epoch schedule is deterministic, the
update sequence is oblivious, and the core module does not use
the augmentation samples or outputs to guide its decisions.

After each update, including its required core repair and any
resulting reset, recompute the two augmentation outputs from
Lemma~\ref{set-lem-augmentation}.  Each extension
optimizes only the current buffer, starting from $R$ or from the
stored core answer $S$.  It takes
$O(|D|+1)=O(d)$ queries.  Store the resulting membership lists,
query the values of $S,T_R,T_S$, and return a handle to the
largest-value candidate.  Perform this comparison at initialization
as well.  Every forced set and every extension element is active,
so all candidates and all query sets are feasible current subsets.
The core profile and the augmentation lemma prove the approximation
at each fixed time.  No conditioning on previous approximation
success is used.

It remains to account for the core rebuilds and their repairs.
The initial module costs at most $n_0q(N,\varepsilon)$, including
every repair until its destruction.  Each subsequent module has
at most $N$ initial elements.  Prepay its entire lifetime bound
$Nq(N,\varepsilon)$ from the $d$ actual updates preceding its
initialization.  There are at most $\lfloor T/d\rfloor$ such
initializations in a prefix of $T$ updates, so this cost is at most
$(TN/d)q(N,\varepsilon)$.  The prepayment also covers a module
that is still alive at the end of the prefix; no charge is made
to future updates.  The two extensions and value comparisons
cost $O(Td+1)$ in total.  This proves
\eqref{set-eq-sqrt-dynamic-total}.

All membership lists are updated before the answer is returned.
Enumerating its members therefore makes no additional oracle queries.
\end{proof}

\subsection{Choice of parameters}
\label{set-sec-root-parameters}

\label{set-sec-sqrt-parameters}

The square-root argument is valid for every choice satisfying
\eqref{set-eq-augmentation-conditions}. Choose
\[
 p=\frac4{11},\quad q=\frac7{11},\quad
 r=\frac5{33},\quad s=\frac{33}{20},\quad
 u=\frac{51}{20},\quad v=\frac5{51}.
\]
Then $rs=uv=1/4$ and
\[
 \lambda=\frac{151}{220},\qquad
 L=\frac{82901}{79860},\qquad A=\frac{371}{220},\qquad
 d_R=\frac{1849}{660},\qquad d_S=\frac{3721}{1020}.
\]
In particular $\lambda,L>0$. For the remaining inequalities, use
$1/\alpha=1+\sqrt2$ and $\beta/\alpha=1/\sqrt2$.
The elementary lower bounds
$\sqrt2>7071/5000$ and $1/\sqrt2>7071/10000$ imply
\[
 \frac L\alpha-\lambda d_S>\frac1{500},
 \qquad
 \frac{\beta}{\alpha}L+rp^2-\lambda(1+v)>\frac1{2500}.
\]
Both comparisons are rational after substituting these lower bounds.
Thus all coefficients discarded in the proof are nonnegative. Finally,
\[
 \rho_{\rm sq}:=\frac{\alpha A}{\alpha d_R+L}
 =0.31772\ldots>0.3177.
\]
For an exact verification of the last inequality, substitute
$\alpha>2071/5000$ into the increasing function
$a\mapsto aA/(ad_R+L)$; the resulting rational value exceeds
$3177/10000$.

Taking $d=\lceil\sqrt N\rceil$ balances rebuilding the decremental core
against optimizing the insertion buffer. The error term in
Lemma~\ref{set-lem-augmentation} has coefficient less than one,
so using core accuracy $\varepsilon$ gives the claimed
$(\rho_{\rm sq}-\varepsilon)$ approximation.

\section{The general approximation-update tradeoff}
\label{set-sec-general-tradeoff}
\label{sec:general-tradeoff}

The square-root algorithm combines a decremental core with a small
insertion buffer. To obtain every fixed positive update exponent,
we replace this partition by a hierarchy of $K$ cohorts with
geometrically spaced rebuilding periods. The decremental mask
structure still applies within each cohort. The remaining task is
to combine the maintained cohort solutions: taking only the best
one would lose a factor of $K$.

The approximation analysis uses the slack of half sampling.
If the slack is large, the half-sample itself gives the desired
improvement. When it is small, paired chains yield nearly optimal
solutions within the cohorts, and their interactions admit a
controlled quadratic approximation. Choosing a uniform set of
$\lfloor K/2\rfloor$ cohorts then captures these interactions with probability strictly
larger than $1/4$.

We first establish the structural estimates, then derive the
decremental cohort guarantee and the balanced composition bound.
We finish with the fully dynamic hierarchy and its parameter choices.

\subsection{Structure near the quarter-sampling bound}
\label{set-sec-clean-structure}
\label{set-sec-round2-foundations}

The same half-sampling slack controls two difficulties. Within one
cohort it bounds the loss of a paired-chain solution. Across cohorts
it bounds the error of representing the objective by interactions
between good and bad elements. We first establish the common
structural estimates.

Throughout this subsection, $f:2^U\to\mathbb R_{\ge0}$ is submodular.
For a set $C$, write $C(p)$ for its independent Bernoulli-$p$ sample.
Fix an optimum $O\subseteq U$ and put
\[
 B=U\setminus O,\quad M=f(O),\quad f_0=f(\varnothing),\quad
 \mu=\mathbb E f(U(1/2)),\quad
 \eta=4\mu-M-f_0-f(U),\quad \rho=\eta-f(B).
\]
The optimum, minima, and these slack quantities are analysis objects;
the algorithm does not compute them.  The quantity $\eta$ measures the slack in the half-sampling bound.
We prove that this slack controls both conditional nonmodularity
and errors caused by assigning elements the wrong optimization role.

\subsubsection{Sampling, modular lower vectors, and residuals}

The multilinear expectation of a submodular function is concave on
every line with a nonnegative direction.  Indeed, its second
derivative on such a line is a nonnegative weighted sum of mixed
second differences, each of which is nonpositive.  Applying this
concavity successively to two disjoint blocks $A,D$ gives
\begin{equation}\label{set-eq-clean-four-corners}
 4\mathbb E h((A\cup D)(1/2))
 \ge h(\varnothing)+h(A)+h(D)+h(A\cup D).
\end{equation}
Here $h$ can be any submodular function; nonnegativity is not needed.
For $A=O,D=B$, this proves $\rho\ge0$ and hence $\eta\ge0$.
For a nonnegative $h$ on $W$, choosing $A$ to maximize $h$ and
discarding the nonnegative term $h(W\setminus A)$ also gives
\begin{equation}\label{set-eq-clean-quarter-endpoints}
 4\mathbb E h(W(1/2))
 \ge \max h+h(\varnothing)+h(W).
\end{equation}

Fix any order $i_1,\ldots,i_m$ of $W$ and define its greedy vector by
\[
 w_{i_j}=h(\{i_1,\ldots,i_j\})-h(\{i_1,\ldots,i_{j-1}\}).
\]
Summing diminishing returns along this order proves
\begin{equation}\label{set-eq-clean-base}
 h(T)\ge h(\varnothing)+w(T)\quad(T\subseteq W),\qquad
 w(W)=h(W)-h(\varnothing).
\end{equation}
In particular, for nonnegative $h$,
\begin{equation}\label{set-eq-clean-base-mass}
 \sum_i(w_i)_+\le\max h-h(\varnothing),\qquad
 \sum_i(-w_i)_+\le\max h-h(W).
\end{equation}
The first inequality evaluates \eqref{set-eq-clean-base} on the positive
support; the second follows by subtracting the total sum of $w$.

Define the midpoint gap and the corresponding residual by
\[
 D(h)=2\mathbb E h(W(1/2))-h(\varnothing)-h(W),\qquad
 u(T)=h(T)-h(\varnothing)-w(T).
\]
For every $T\subseteq W$, submodularity and
\eqref{set-eq-clean-four-corners} imply
\[
 0\le h(T)+h(W\setminus T)-h(\varnothing)-h(W)\le2D(h).
\]
Also \eqref{set-eq-clean-base} on $W\setminus T$ gives
$w(T)\ge h(W)-h(W\setminus T)$.  Consequently
\begin{equation}\label{set-eq-clean-residual}
 u\ge0,\quad u\text{ is submodular},\quad
 u(\varnothing)=u(W)=0,\quad
 u(T)\le2D(h),\quad \mathbb E u(W(1/2))=D(h)/2.
\end{equation}
Concavity and the zero endpoints imply
\begin{equation}\label{set-eq-clean-residual-expectation}
 \mathbb E u(W(p))\le D(h)\max\{p,1-p\}.
\end{equation}
For example, when $p\le1/2$, express $1/2$ as the convex
combination of $p$ and $1$ with coefficient $1/[2(1-p)]$ on $p$.
Concavity then gives \eqref{set-eq-clean-residual-expectation}; the other
case is symmetric.

We shall also need a dimension-free bound on residual derivatives.
For any nonnegative submodular $u$, the sum of its positive deletion
marginals inside a set $S$ is at most $u(S)$: delete just the elements
with positive deletion marginal, using diminishing returns, and use
nonnegativity of the resulting value.  Likewise, the sum of the
absolute negative addition marginals outside $S$ is at most $u(S)$:
add the elements with negative marginal and use nonnegativity of
the final value.  Taking expectations for $S\sim W(p)$ and using
independence of each element's own sampling bit yields
\[
 \sum_{i\in W}\mathbb E\left|u(i\mid (W\setminus\{i\})(p))\right|
 \le\frac{\mathbb E u(W(p))}{p(1-p)}.
\]
For the residual in \eqref{set-eq-clean-residual}, this becomes
\begin{equation}\label{set-eq-clean-residual-gradient}
 \sum_{i\in W}\mathbb E\left|u(i\mid (W\setminus\{i\})(p))\right|
 \le\frac{D(h)}{\min\{p,1-p\}}.
\end{equation}

\subsubsection{Wrong-role budgets and complementary conditional gaps}
\label{set-sec-wrong-role-budgets}

Set
\[
 x=\mathbb E f(O(1/2)),\quad y=\mathbb E f(B(1/2)),\quad
 u=\mathbb E f(O\cup B(1/2)),\quad
 v=\mathbb E f(B\cup O(1/2)).
\]
Concavity gives $\mu\ge(y+u)/2$, $\mu\ge(x+v)/2$,
$u\ge(M+f(U))/2$, and $x\ge(f_0+M)/2$.  Thus
$y\le(f_0+\eta)/2$ and $v\le(f(U)+\eta)/2$.
Applying \eqref{set-eq-clean-quarter-endpoints} to $f|_B$ and to
$X\mapsto f(B\cup X)$ on $O$ proves
\begin{align}
 r&:=\max_{T\subseteq B}f(T)-f_0
       \le4y-2f_0-f(B)\le2\eta-f(B),
       \label{set-eq-clean-bad-budget}\\
 r_k&:=\max_{T\subseteq O}f(B\cup T)-f(U)
       \le4v-f(B)-2f(U)\le2\eta-f(B).
       \label{set-eq-clean-good-budget}
\end{align}
Both budgets are nonnegative.  Adding any subset $Z\subseteq B$
to a disjoint context increases the value by at most $r$, since its
addition marginal is at most $f(Z)-f_0$.  Adding $H\subseteq O$
to a disjoint context contained in $U$ decreases the value by at
most $r_k$, since its smallest addition marginal is
$f(U)-f(U\setminus H)\ge-r_k$.

For $Z\subseteq B$, let $D_Z$ be the midpoint gap of
$X\mapsto f(X\cup Z)$ on $O$, and put $Z'=B\setminus Z$.
For each $X\subseteq O$, applying \eqref{set-eq-clean-four-corners}
to the split $Z,Z'$ of the $B$ coordinates gives
\[
 f(X\cup Z)+f(X\cup Z')
 \le4\mathbb E f(X\cup B(1/2))-f(X)-f(X\cup B).
\]
Average over $X\sim O(1/2)$.  At the two $O$ endpoints,
submodularity bounds the sum of the four subtracted values below
by $f_0+f(B)+M+f(U)$.  Hence
\begin{align*}
 D_Z+D_{Z'}
 &\le8\mu-2x-2v-[f_0+f(B)+M+f(U)]\\
 &\le2\rho.
\end{align*}
The last step uses $2x\ge f_0+M$ and $2v\ge f(B)+f(U)$.
Moreover, $\mu\ge(x+v)/2$ gives
\begin{equation}\label{set-eq-clean-conditional-gaps}
 D_Z+D_{Z'}\le2\rho,\qquad D_{\varnothing}+D_B\le\rho.
\end{equation}
The proof is symmetric in the blocks $O,B$, so both assertions
also hold with their roles interchanged.  In particular, the
midpoint gap of either whole color block in every fixed context
on the other block is at most $2\rho\le2\eta$.

\subsubsection{Centered slack of an arbitrary conditional subproblem}

\begin{lemma}[Restriction and conditioning of a midpoint gap]
For any submodular $h$ on $W$, partition $W=R\mathbin{\dot\cup}
P\mathbin{\dot\cup}Q$ and define $h_R(T)=h(R\cup T)$ on $P$.
Then $D(h_R)\le2D(h)$.  More precisely, for every split
$P=A\mathbin{\dot\cup}D$,
\begin{equation}\label{set-eq-clean-conditioned-interaction}
 h(R\cup A)+h(R\cup D)-h(R)-h(R\cup P)\le2D(h).
\end{equation}
\end{lemma}
\begin{proof}
For a split $P=A\mathbin{\dot\cup}D$, submodularity twice gives
\begin{align*}
 &h(R\cup A)+h(R\cup D)-h(R)-h(R\cup P)\\
 &\qquad\le h(R\cup A)+h(D\cup Q)-h(\varnothing)-h(W).
\end{align*}
To verify the direction, the right side minus the left side is
\[
 h(D\cup Q)+h(R)+h(R\cup P)-h(\varnothing)-h(W)-h(R\cup D).
\]
The first two terms are at least $h(R\cup D\cup Q)+h(\varnothing)$;
submodularity on $R\cup D\cup Q$ and $R\cup P$ makes the remainder
nonnegative.  The right side of the first display is a global
complementary-pair slack and is at most $2D(h)$ by
\eqref{set-eq-clean-residual}.  Averaging over a fair split of $P$
proves the lemma.
\end{proof}

\begin{theorem}[Conditional centered slack]
\label{set-thm-clean-conditional-slack}
For arbitrary disjoint $R,C\subseteq U$, define $g(T)=f(R\cup T)$
on $C$.  Let
\[
 m_g=\min_{T\subseteq C}g(T),\qquad
 \eta_g=4\mathbb E g(C(1/2))-\max g-g(\varnothing)-g(C).
\]
Then
\begin{equation}\label{set-eq-clean-twenty}
 0\le\eta_g-m_g\le C_{\mathrm c}\eta.
\end{equation}
The assertion is pointwise for all choices of $R,C$, with a
constant $C_{\mathrm c}$ fixed in Subsection~\ref{set-sec-tradeoff-parameters}.
\end{theorem}
\begin{proof}
The lower bound follows by applying
\eqref{set-eq-clean-quarter-endpoints} to $g-m_g$.
Put $P=C\cap O$ and $Q=C\cap B$.  Each midpoint gap obtained by
varying $P$ or $Q$ with the other coordinates fixed is at most
$4\eta$: apply the preceding lemma inside its whole color block
and then use \eqref{set-eq-clean-conditional-gaps}.

For $Z\subseteq Q$ write
\[
 d_P(Z)=2\mathbb E_{X\sim P(1/2)}g(X\cup Z)-g(Z)-g(P\cup Z),
\]
and write $d_Q(\varnothing),d_Q(P)$ for the analogous gaps on $Q$
with the $P$ coordinates respectively empty and full.  Expanding
all terms gives the exact identity
\begin{align*}
 \Gamma&:=4\mathbb E g(C(1/2))-g(P)-g(Q)-g(\varnothing)-g(C)\\
   &=2\mathbb E_{Z\sim Q(1/2)}d_P(Z)
                         +d_Q(\varnothing)+d_Q(P)\le16\eta.
\end{align*}
Every $T\subseteq C$ also satisfies $g(T)\ge g(Q)-r-r_k$:
add its missing $Q$ coordinates, whose possible increase is at
most $r$, and then remove its $O$ coordinates, whose possible
increase is at most $r_k$.  The resulting set is $R\cup Q$.
Thus $m_g\ge g(Q)-r-r_k$, whereas $\max g\ge g(P)$.
It follows that $\eta_g-m_g\le\Gamma+r+r_k\le20\eta$.
\end{proof}

The proof uses a four-corner error of at most $16\eta$ and
wrong-role budgets totaling at most $4\eta$. This verifies the
choice of $C_{\mathrm c}$ made in
Subsection~\ref{set-sec-tradeoff-parameters}.

\subsection{Paired chains on near-tight instances}
\label{set-sec-clean-paired-chains}
\label{set-sec-paired-chain-hard-anchor}

A single random chain can underestimate the optimum even when its
positive mask has high value. For the fixed-depth construction, we
need a different guarantee: the maintained mask should approach the
optimum as the half-sampling slack tends to zero.

We average forward and reverse chain marginals. Their conditioning
sets are complementary, so the conditional midpoint bounds from the
previous section control their combined error. Marginals near the
ends of an order are less stable. We therefore retain only a central
interval, and use a second order shifted by one half so that every
element is central in exactly one of the two orders.

Assign independent uniform priorities $\pi_i\in[0,1)$ to the
elements of $U$.  Let $w_i^{\to},w_i^{\leftarrow}$ be their raw
greedy marginals in increasing and decreasing priority order, and
write $z_i=(w_i^{\to}+w_i^{\leftarrow})/2$.
This averaged vector satisfies \eqref{set-eq-clean-base}.  Its interval
sums need four prefix/suffix value queries.  Inserting an element
decreases every old coordinate; deleting one increases every
surviving coordinate, by diminishing returns.

\subsubsection{A fixed-base comparison for the paired coordinates}

Use one fixed deterministic order on $O$.  For each $Z\subseteq B$,
let $m^Z$ and $u_Z$ be the greedy vector and residual of
$X\mapsto f(X\cup Z)$ in this order.  For $Z'=B\setminus Z$, set
\[
 J_Z(X)=f(X\cup Z)+f(X\cup Z')-f(X)-f(X\cup B).
\]
Submodularity gives $J_Z(X)\ge0$.  The complementary-pair bound
for the function on $B$ conditioned on $X$, together with
\eqref{set-eq-clean-conditional-gaps}, gives $J_Z(X)\le4\rho$.
Its decomposition into the four greedy vectors is
\[
 J_Z(X)=J_Z(\varnothing)+c^Z(X)
       +u_Z(X)+u_{Z'}(X)-u_{\varnothing}(X)-u_B(X),
\]
where $c^Z=m^Z+m^{Z'}-m^{\varnothing}-m^B$.
By \eqref{set-eq-clean-residual} and
\eqref{set-eq-clean-conditional-gaps}, the residual difference belongs
to $[-2\rho,4\rho]$.  Evaluate the display on the positive and
negative supports of $c^Z$.  Its positive mass is at most
$6\rho-J_Z(\varnothing)$ and its negative mass is at most
$4\rho+J_Z(\varnothing)$.  Hence
\begin{equation}\label{set-eq-clean-paired-lone}
 \|m^Z+m^{Z'}-m^{\varnothing}-m^B\|_1\le10\rho.
\end{equation}

Put $v=m^{\varnothing}$ and $u=m^B$.  The vector $v$ is
nonnegative: optimality gives $f(i\mid O\setminus\{i\})\ge0$
for $i\in O$, and diminishing returns makes every earlier
$O$-only marginal nonnegative as well.  Its total is
$v(O)=M-f_0$.  The negative mass of $u$ is at most $r_k$ by
\eqref{set-eq-clean-base-mass}.  Condition on one coordinate's priority
$\pi_i=p$.  Its forward prefix is a Bernoulli-$p$ sample of the
other coordinates; its reverse prefix is the complementary sample.
Couple these prefixes as $R_O\sim O(p)$ and $Z\sim B(p)$, omitting
the differentiated element from either marginal.  The exact
decompositions into bases and residuals give
\[
 2z_i=m_i^Z+m_i^{B\setminus Z}
      +u_Z(i\mid R_O\setminus\{i\})
      +u_{B\setminus Z}(i\mid O\setminus(R_O\cup\{i\})).
\]
Consequently, if $e_i=2z_i-v_i-u_i$ and
$d_p=\min\{p,1-p\}$, then
\begin{equation}\label{set-eq-clean-coordinate-error}
 \sum_{i\in O}\mathbb E[|e_i|\mid\pi_i=p]
 \le10\rho+\frac{2\rho}{d_p}.
\end{equation}
Indeed, \eqref{set-eq-clean-paired-lone} bounds the base error, and
\eqref{set-eq-clean-residual-gradient} bounds the two residual derivative
sums by $(D_Z+D_{B\setminus Z})/d_p\le2\rho/d_p$.
In the left side each coordinate is conditioned on its own priority
being $p$; linearity of expectation permits the common auxiliary
sampling coupling used to bound that sum.

Apply the same construction on $B$, using a fixed deterministic
order there.  Let $b$ be the base of $f|_B$, and let $-t$ be the
base of $Z\mapsto f(O\cup Z)$.  Global optimality of $O$ and
diminishing returns imply $t_i\ge0$ for every $i\in B$.
The positive mass of $b$ is at most $r$.  The symmetric version of
\eqref{set-eq-clean-coordinate-error} reads $2z_i=b_i-t_i+e_i$ on $B$,
with the same summed error bound.

These deterministic bases also give, for every set $A\subseteq U$,
\begin{equation}\label{set-eq-clean-intrinsic-loss}
 f(A)\ge M-v(O\setminus A)-t(A\cap B).
\end{equation}
To prove it, put $K=A\cap O$ and $Z=A\cap B$.  The base on $O$
gives $f(K)\ge f_0+v(K)$, and the base on $B$ conditioned on $O$
gives $f(O\cup Z)\ge M-t(Z)$.  Diminishing returns implies
$f(K\cup Z)-f(K)\ge f(O\cup Z)-f(O)$; combine the three
inequalities and use $v(O)=M-f_0$.

\subsubsection{Interior sign errors and the two-shift theorem}

Fix $0<a<1/2$, put $\ell=1-2a$, and let $\widetilde A$ be any
random mask contained in the priority interval $[a,1-a)$.
Its weighted sign error is
\[
 E_a=\sum_{i:\,a\le\pi_i<1-a}|z_i|
  \mathbf1\{\mathbf1_{i\in\widetilde A}\ne\mathbf1_{z_i>0}\}.
\]
The mask may depend on all priorities and all queries.
For an omitted good coordinate in the interval,
$v_i\le(u_i)_-+|e_i|+2(z_i)_+$; the final term is needed only
when the sign decision is wrong.  Similarly, an included bad
coordinate has $t_i\le(b_i)_++|e_i|+2(-z_i)_+$, again with the
last term needed only for a wrong sign.  Integrating
\eqref{set-eq-clean-coordinate-error} over $p\in[a,1-a]$ gives
\begin{align}
 \mathbb E v(\{i\in O:a\le\pi_i<1-a,\ i\notin\widetilde A\})
 &\le\ell r_k+
       \left(10\ell+4\log\frac1{2a}\right)\rho
       +2\mathbb E E_a^{O},\label{set-eq-clean-good-interior}\\
 \mathbb E t(\widetilde A\cap B)
 &\le\ell r+
       \left(10\ell+4\log\frac1{2a}\right)\rho
       +2\mathbb E E_a^{B},\label{set-eq-clean-bad-interior}
\end{align}
where $E_a^O,E_a^B$ are the corresponding parts of $E_a$.
The integral used here is
$\int_a^{1-a}2/d_p\,dp=4\log(1/(2a))$.

\begin{theorem}[Two shifted paired-chain masks]
\label{set-thm-clean-two-shift}
Let $\pi'_i=(\pi_i+1/2)\bmod1$.  Form the paired vectors for the
two orders, and let $\widetilde A_1,\widetilde A_2$ be maintained
masks contained in their respective intervals $[1/4,3/4)$.
Let $E_1,E_2$ be their weighted sign errors as defined above.
Then their feasible union $S=\widetilde A_1\cup\widetilde A_2$
satisfies
\begin{equation}\label{set-eq-round2-two-shift-profile}
 \mathbb E f(S)\ge M-C_{\mathrm p}\eta-2\mathbb E(E_1+E_2).
\end{equation}
If $m=\min_{T\subseteq U}f(T)$, the stronger centered form is
\begin{equation}\label{set-eq-clean-two-shift-centered}
 \mathbb E f(S)\ge M-C_{\mathrm p}(\eta-m)-2\mathbb E(E_1+E_2).
\end{equation}
Here $C_{\mathrm p}$ is the constant fixed in Subsection~\ref{set-sec-tradeoff-parameters}. The current universe must be fixed independently of the sampled
priorities; no independence between the two masks is required.
\end{theorem}
\begin{proof}
Each priority family is independently uniform across elements.
The two interval sets partition $U$: every identity is interior in
exactly one order under the half-open interval convention.  Use the
same fixed vectors $v,t$ in the two applications of
\eqref{set-eq-clean-good-interior}--\eqref{set-eq-clean-bad-interior}.
An omitted good coordinate is misclassified in its eligible order,
and the $t$ mass of the union of included bad coordinates is at most
the sum of the two $t$ masses.  With $a=1/4$ and $\ell=1/2$, the
combined expected intrinsic loss is therefore at most
\[
 r+r_k+(20+16\ln2)\rho+2\mathbb E(E_1+E_2)
 \le C_{\mathrm p}\eta+2\mathbb E(E_1+E_2).
\]
Apply \eqref{set-eq-clean-intrinsic-loss}.  For the centered form,
replace $f$ by $f-m$.  All marginal vectors, signs, and errors stay
unchanged, whereas the optimum becomes $M-m$ and the endpoint slack
becomes $\eta-m$.  Adding $m$ back proves the claim.
\end{proof}

Finally, apply the centered theorem to a conditional objective
$g(T)=f(R\cup T)$ and use Theorem~\ref{set-thm-clean-conditional-slack}.
For any two maintained interior masks on this subproblem,
\begin{equation}\label{set-eq-clean-conditional-paired-profile}
 \mathbb E g(S)\ge\operatorname{OPT}(g)
                    -C_{\mathrm c}C_{\mathrm p}\eta-2\mathbb E(E_1+E_2).
\end{equation}
When the subproblem and its outside context are determined before
its fresh priorities are drawn, this is a conditional guarantee;
averaging gives the same statement for such random subproblems.
The two structural results above supply the constants used in the
fully dynamic cohort analysis.

The proof requires $C_{\mathrm p}\ge24+16\ln2$, precisely the
choice in Subsection~\ref{set-sec-tradeoff-parameters}. The two
half-shifted eligibility intervals partition the universe; thus an
element is never excluded merely because it is extreme in one order.

\subsection{Maintained cohort solutions}
\label{set-sec-cohort-proof}

We now combine the paired-chain estimate with the decremental mask
structure. The resulting module maintains a cohort solution whose
loss is controlled by the half-sampling slack of the current larger
universe. This form will allow the cohort guarantees to be summed
inside the dynamic hierarchy.

For a current universe $U$, define its half-sampling slack by
\[
 \eta_U=4\E f(U(1/2))-\OPT(U)-f(\varnothing)-f(U)\ge0.
\]
Here $U(1/2)$ is an independent half-sample. The inequality follows
from the four-corner sampling bound in
Lemma~\ref{lem:prelim-sampling}.
The algorithm never evaluates $\eta_U$.

We use the following explicit cohort procedure. At each value-epoch
start, draw independent uniform priorities $\pi_i$ and the shifted
priorities $\pi'_i=(\pi_i+1/2)\bmod1$. In each order, average its
forward and reverse chain marginals and apply
Lemma~\ref{set-lem-mask-maintenance}, retaining only priorities in
$[1/4,3/4)$. Return the union of the two maintained masks. Each
identity is eligible in exactly one order.

\begin{corollary}[Maintained cohort solutions]
\label{set-cor-cohort}
Fix $0<\xi<1/20$. Suppose a decremental cohort and its deletion sequence are fixed
independently of the module's randomness. At any fixed time, let
$U\supseteq C$ be a larger universe, also independent of that randomness.
There is a stored set $S_C\subseteq C$ with
\begin{equation}\label{set-eq-cohort-deficit}
 \E f(S_C)\ge
 \OPT(C)-D_0\eta_U-C_{\mathrm e}\xi\OPT(U).
\end{equation}
Here $D_0,C_{\mathrm e}$ are absolute constants specified in
Subsection~\ref{set-sec-tradeoff-parameters}. Its initialization
and every deletion prefix use at most $m q(N,\xi)$ queries when
the initial cohort has size $m$.
\end{corollary}

\begin{proof}[Proof of Corollary~\ref{set-cor-cohort}]
Use Lemma~\ref{set-lem-mask-maintenance} for the two shifted averaged
chain vectors, with their respective central eligibility intervals.
Draw fresh independent uniform priorities at each value-epoch start,
and obtain the second order by the half-shift. The retained survivors
at every fixed time are independent of these priorities.
The two missing-mass bounds give
$2\E(E_1+E_2)\le8\xi\OPT(W)$. The conditional centered-slack theorem
applies to the current restriction $W\subseteq U$, and truncation loses
at most $2\xi\OPT(C)$. Thus the paired-chain estimate gives
\[
 \E f(S_C)\ge
 \OPT(C)-C_{\mathrm c}C_{\mathrm p}\eta_U-10\xi\OPT(U).
\]
We may therefore take $C_{\mathrm e}=10$. The query bound is
\eqref{set-eq-mask-cost}. The bound concerns the slack of the
current larger universe, rather than the state at which the module
was created.
\end{proof}

\subsection{Balanced composition of cohort solutions}
\label{set-sec-clean-balanced-cohorts}

Independent fair sampling of cohort solutions retains an interaction
from one selected cohort to another unselected cohort with probability
$1/4$. Choosing exactly half of the cohort indices increases this
probability to $q_K>1/4$. To turn this elementary gain into an
approximation improvement, we must show that the useful part of the
objective is captured by interactions between different cohorts.
Near-tightness of half sampling gives precisely this statement.

The proof replaces each cohort's good and bad portions by one coarse
coordinate. Up to an additive error controlled by $\eta$, the resulting
function is a bipartite quadratic polynomial. Same-cohort interactions
are charged to the local optimization deficits; balanced sampling then
captures the remaining interactions at rate $q_K$.

Let $f:2^U\to\mathbb R_{\ge0}$ be submodular, fix an optimum $O$,
and put $B=U\setminus O$, $M=f(O)$, $f_0=f(\varnothing)$, and
\[
 \eta=4\mathbb E f(U(1/2))-M-f_0-f(U).
\]
Write
\[
 r_B=\operatorname{OPT}(B)-f_0\le2\eta,\qquad
 r_k=\max_{X\subseteq O}f(B\cup X)-f(U)\le2\eta,
\]
and $D_O=2\mathbb E f(O(1/2))-f_0-M\le\eta$.
We use the following proved midpoint facts: every conditional function on one of $O,B$, with the other color fixed, has midpoint gap at most $2\eta$; a restriction with a fixed context has gap at most twice the original gap; and a greedy-base residual lies in $[0,2D]$ when its midpoint gap is $D$.  These facts hold for real-valued submodular functions and do not require normalization.

Relabel the cohort slots by $1,\ldots,K$ in this section.
Let $C_1,\ldots,C_K$ partition $U$, where $K\ge2$, and let
$S_j\subseteq C_j$ be arbitrary masks.  Empty cohorts are allowed.
Define their local deficits by
\[
 \delta_j=\operatorname{OPT}(C_j)-f(S_j)\ge0.
\]
Choose a uniformly random set $J\subseteq[K]$ of cardinality
$s=\lfloor K/2\rfloor$, independently of the masks, and return
$S_J=\bigcup_{j\in J}S_j$.  Put
\[
 p_K=\frac sK,\qquad
 q_K=\frac{s(K-s)}{K(K-1)}
 =\begin{cases}
 K/(4(K-1)),&K\text{ even},\\
 (K+1)/(4K),&K\text{ odd}.
 \end{cases}
\]
In particular, $q_K>1/4$ for every fixed finite $K$.

\begin{theorem}[Balanced cohort composition]
\label{set-thm-balanced}
For every realization of the masks,
\begin{equation}\label{set-eq-clean-balanced-cohort-main}
 \mathbb E_J f(S_J)
 \ge q_K M-2q_K\sum_{j=1}^K\delta_j-C_K\eta,
\end{equation}
where $C_K$ is a coefficient determined at the end of the proof.
In particular $C_K=O(K)$.
The assertion permits the masks to depend on arbitrary preprocessing
randomness.  Only the final balanced subset $J$ must be independent
of those masks.
\end{theorem}

We prove the theorem in the following three steps.

\subsubsection{A coarse bipartite quadratic approximation}

Write $G_j=S_j\cap O$, $H_j=S_j\cap B$, and let $G=\bigcup_jG_j$,
$H=\bigcup_jH_j$.  Use the fixed group order $1,\ldots,K$ and define
\begin{align*}
 a_i(Z)&=f(G_{<i}\cup G_i\cup Z)-f(G_{<i}\cup Z),\\
 a_i&=a_i(\varnothing),\\
 b_j&=f(H_{<j}\cup H_j)-f(H_{<j}),\\
 c_{ij}&=a_i(H_{<j})-a_i(H_{<j}\cup H_j)\ge0.
\end{align*}
Here $G_{<i}=\bigcup_{s<i}G_s$, and similarly for $H_{<j}$.
For $X,Y\subseteq[K]$, define
\[
 Q(X,Y)=f_0+\sum_{i\in X}a_i+\sum_{j\in Y}b_j
                    -\sum_{i\in X,\,j\in Y}c_{ij}.
\]
This polynomial is used only for analysis; the algorithm need not
construct its coefficients.

\begin{lemma}[Uniform lower approximation on the coarse groups]
\label{set-lem-coarse-lower}
For all $X,Y\subseteq[K]$,
\[
 f(G_X\cup H_Y)\ge Q(X,Y)-8K\eta.
\]
Moreover, $Q([K],\varnothing)=f(G)$ and
$\sum_jb_j=f(H)-f_0$ exactly.
\end{lemma}
\begin{proof}
For fixed $Z\subseteq B$, the midpoint gap of the function on all of
$O$ obtained by conditioning on $Z$ is at most $2\eta$.
Restricting this function to a subset of $O$ increases its midpoint
gap by at most a factor two, by the conditioned-gap lemma.
Replacing independent half sampling within each group by one fair
all-or-none group coin can only decrease the expectation, by
submodularity and diagonal concavity.  Thus the group function on
the $G_i$ has midpoint gap at most $4\eta$.
Its fixed-order greedy-base residual is nonnegative and at most
$8\eta$.

The same reasoning applies to each function on the $H_j$ obtained by
fixing an arbitrary subset of $O$.  Hence, for each $i$,
the two functions $Y\mapsto f(G_{\le i}\cup H_Y)$ and
$Y\mapsto f(G_{<i}\cup H_Y)$ have greedy-base residuals in
$[0,8\eta]$.  Subtracting their base expansions gives
\[
 a_i(H_Y)=a_i-\sum_{j\in Y}c_{ij}+e_i(Y),
 \qquad |e_i(Y)|\le8\eta.
\]
The base expansion of $f|_H$ has a nonnegative residual, as does the
base expansion on the $G_i$ at fixed context $H_Y$.
Substituting the preceding identity into these two expansions gives
the stated lower bound.  For the two exact identities, all groups
are taken in their entirety, so the base sums telescope.
\end{proof}

Define $u_i=a_i-\sum_jc_{ij}=a_i(H)$.  Filling the conditioning set
$H$ to all of $B$ only decreases these marginals.  The resulting
group weights are sums of consecutive coordinates of a greedy base
of $X\mapsto f(B\cup X)$ on $O$, with each $G_i$ kept consecutive
and all of $G$ ordered before $O\setminus G$.
The negative part of a group sum is at most the sum of the negative
parts of its coordinates.  The latter total is at most $r_k$.
Consequently
\begin{equation}\label{set-eq-clean-coarse-negative-budget}
 \sum_i(u_i)_-\le r_k.
\end{equation}

\subsubsection{Omitted good mass and same-cohort interactions}

Fix a greedy base $v$ of $f|_O$.  All its coefficients are nonnegative,
and its residual is between zero and $2D_O$.  With
$P_j=C_j\cap O$, subadditivity gives
\[
 f(P_j)-f(G_j)
 \le\delta_j+f(H_j)-f_0\le\delta_j+r_B.
\]
Therefore
\[
 v(P_j\setminus G_j)\le\delta_j+r_B+2D_O,
\]
and summing yields
\begin{equation}\label{set-eq-clean-retained-good-value}
 f(G)\ge M-\sum_j\delta_j-K(r_B+2D_O).
\end{equation}

We also need to control the diagonal interactions $c_{jj}$.  For
disjoint sets $A,D$ and a disjoint context $Z$, write
\[
 I(A,D\mid Z)=f(Z\cup A)+f(Z\cup D)-f(Z)-f(Z\cup A\cup D).
\]
Every such interaction between two subsets of $O$, with a fixed
$B$ context, is at most $4\eta$; the analogous assertion holds on
$B$ with a fixed $O$ context.  Indeed, the interaction embeds into
a complementary-pair gap of the full conditional function, whose
midpoint gap is at most $2\eta$.

For $A\subseteq O$, $D\subseteq B$, $Z_O\subseteq O\setminus A$,
and $Z_B\subseteq B\setminus D$, the exact interaction identities
give
\begin{align*}
 I(A,D\mid Z_O\cup Z_B)-I(A,D\mid Z_B)
 &=I(A,Z_O\mid Z_B\cup D)-I(A,Z_O\mid Z_B),\\
 I(A,D\mid Z_B)-I(A,D\mid\varnothing)
 &=I(D,Z_B\mid A)-I(D,Z_B\mid\varnothing).
\end{align*}
Each right-hand side is a difference of two numbers in
$[0,4\eta]$.  Hence conditioning changes the cross interaction by
at most $8\eta$ in absolute value.  In particular,
\[
 c_{jj}\le I(G_j,H_j\mid\varnothing)+8\eta
 \le\delta_j+r_B+8\eta,
\]
where the second inequality uses $f(G_j)\le f(P_j)$.
Thus
\begin{equation}\label{set-eq-clean-diagonal-budget}
 \sum_jc_{jj}\le\sum_j\delta_j+K(r_B+8\eta).
\end{equation}

\subsubsection{Balanced sampling captures every off-diagonal interaction}

Rewrite the coarse polynomial as
\[
 Q(X,Y)=f_0+\sum_{i\in X}u_i+\sum_{j\in Y}b_j
                         +\sum_{i\in X,\,j\notin Y}c_{ij}.
\]
Under a uniformly random $s$-subset $J$, each coordinate is selected
with probability $p_K$, and, for $i\ne j$,
$\Pr(i\in J,j\notin J)=q_K$.
Let $C_{\rm diag}=\sum_jc_{jj}$.  The exact identities above give
\begin{align*}
 \mathbb E_J Q(J,J)-q_K(f(G)-C_{\rm diag})
 &=(1-q_K-p_K)f_0+p_K f(H)
                     +(p_K-q_K)\sum_i u_i\\
 &\ge-(p_K-q_K)r_k.
\end{align*}
Here $p_K\ge q_K$ and $p_K+q_K\le1$; nonnegativity of $f_0,f(H)$
and \eqref{set-eq-clean-coarse-negative-budget} justify the inequality.

Apply the coarse lower approximation and then
\eqref{set-eq-clean-retained-good-value} and
\eqref{set-eq-clean-diagonal-budget}.  This yields
\begin{align*}
 \mathbb E_J f(S_J)
 &\ge q_K M-2q_K\sum_j\delta_j\\
 &\quad-q_KK(2r_B+2D_O+8\eta)
       -(p_K-q_K)r_k-8K\eta.
\end{align*}
Using $r_B,r_k\le2\eta$ and $D_O\le\eta$ proves \eqref{set-eq-clean-balanced-cohort-main} with
\[
 C_K=K(8+14q_K)+2(p_K-q_K)\le15K+1.
\]

\subsection{A fully dynamic hierarchy of cohorts}
\label{set-sec-clean-dynamic-schedule}

Fix an integer $K\ge2$, set
\[
 r=\lceil N^{1/K}\rceil,\qquad
 p_j=r^{K-1-j}\quad(0\le j\le K-2),
\]
and choose an accuracy parameter $\xi>0$ to be fixed below.
Maintain $K$ disjoint cohort slots $C_0,\ldots,C_{K-1}$.
The first $K-1$ slots use the decremental module. After each update,
rebuild a fresh module on the bottom slot and retain its current
answer.
The schedule depends on update counts and identities alone.

\paragraph{Initialization and updates.}
Put all initially active lifetimes in $C_0$ and initialize its
module; all lower slots start empty. Before each update, discard
the bottom module but retain its cohort membership.
New insertions enter $C_{K-1}$. A deletion removes the lifetime
from its own slot and invokes its decremental module if the slot
is an upper one. Reinsertion creates a new lifetime in the bottom slot;
it does not reinsert into an old decremental module.

Process an update before the reset due at its integer update index.
At every multiple of $p_j$, level $j$ is due for reset.  If several
levels are due, execute only the highest one, meaning the smallest
index $j$.  Replace $C_j$ by the current active union of slots
$j,j+1,\ldots,K-1$, initialize one new decremental module on that
union, and empty all lower slots.  Discard every old module whose
slot was included in the reset.  Skip empty initializations.
Finally initialize a fresh module on the current bottom cohort,
unless that cohort is empty.

The periods divide one another. Consequently, every upper module sees
only deletions between its initialization and destruction; the bottom
module is rebuilt from its updated cohort. The current cohorts always
partition the active universe. Level zero
has size at most $N\le r^K$.  For $1\le j\le K-2$, every identity
at a level-$j$ initialization was inserted after the last reset of
level $j-1$, so its initialization size is at most
$p_{j-1}=r p_j$.  The bottom slot contains at most $r$ lifetimes.

\paragraph{Maintaining an actual output set.}
At initialization draw one uniform subset $J$ of
$\{0,\ldots,K-1\}$ of size $\lfloor K/2\rfloor$, independently
of all module randomness.  Reuse this $J$ throughout the stream.
Maintain the explicit union
\[
 T_t=\bigcup_{j\in J}S_{C_j,t}.
\]
The schedule and modules never use $J$ or the returned answer to
guide their work, so $J$ remains independent of all current cohort
masks at each fixed time.  Their disjointness allows the union to
be updated directly when a mask changes or a slot is reset.

Independently give each initial lifetime, and then each newly
inserted lifetime, a fair bit and maintain the corresponding current
half-sample $R_t\subseteq U_t$.  At initialization and after every
update, query $f(T_t)$ and $f(R_t)$ and return the stored set of
larger value.  The answer is a feasible set, rather than merely an
estimate of its value.  A handle identifies its stored membership
list; listing that set requires no further oracle queries.  The
fixed choice of $J$ also avoids changing a large union solely
because fresh output-selection randomness was drawn.

\begin{theorem}[Scheduling cost and current profiles]
\label{set-thm-clean-schedule}
For every finite sequence of $T$ arbitrary insertions and deletions,
starting with $n_0$ active elements, the preceding implementation
uses
\begin{equation}\label{set-eq-clean-total-query-cost}
 O\bigl(1+(n_0+TKr)q(N,\xi)\bigr)
\end{equation}
value queries.  The accounting is deterministic, for every outcome
of its private randomness.  At every fixed time of an oblivious
update sequence, its current masks satisfy
\begin{equation}\label{set-eq-clean-total-deficit}
 \mathbb E\sum_{j=0}^{K-1}
 [\operatorname{OPT}(C_j)-f(S_{C_j})]
 \le KD_0\eta_{U_t}+KC_{\mathrm e}\xi\operatorname{OPT}(U_t).
\end{equation}
\end{theorem}
\begin{proof}
Prepay the full allowed cost of every upper module at its
initialization, using Corollary~\ref{set-cor-cohort}. For
$1\le j\le K-2$, a nonempty level-$j$ initialization costs at most
$p_{j-1}q=r p_jq$ and is charged to the preceding $p_j$ actual
updates.  The first such initialization after a higher reset also
has $p_j$ actual preceding updates; the empty slot at the higher
boundary costs nothing.  Level-zero initializations after time
zero cost at most $Nq\le r p_0q$ and are charged to the preceding
$p_0$ updates.  Only the initial level-zero module is preprocessing.
No fictitious future updates are charged when the horizon ends
partway through a period.  Early destruction of a module is
harmless because its cost bound includes every deletion prefix.

Summing the rate $rq$ over the maintained levels, adding the
$O(rq)$ fresh bottom module and two output-value queries per update,
proves
\eqref{set-eq-clean-total-query-cost}.  Each maintained mask obeys
\eqref{set-eq-cohort-deficit}, including the fresh bottom module.
Its priorities are drawn independently of the fixed update sequence
and of the balanced subset $J$. The current cohort partition is
determined by identities and update counts, so the same profile
applies to every slot.
Summing proves \eqref{set-eq-clean-total-deficit}.  The accuracy parameter will be chosen after the approximation analysis.
\end{proof}

All randomness required in these statements is internal to the
algorithm.  No deletion times, future lifetimes, or random ordering
of the input stream are assumed.  The only oracle calls are those
specified in the modules and the two
explicit current output candidates.  The preceding balanced-composition theorem applies to \eqref{set-eq-clean-total-deficit} and the independent
fixed $J$ to give the approximation factor below.

\subsubsection{Approximation and parameter choices}
\label{set-sec-tradeoff-parameters}

At a fixed time, condition on the cohort masks and apply the balanced
composition theorem. Corollary~\ref{set-cor-cohort} and the scheduling
bound give
\[
 \E f(S_J)\ge A_KM-B_K\eta,\qquad
 A_K=q_K(1-2C_{\mathrm e}K\xi),\qquad
 B_K=C_K+2q_KKD_0.
\]
The independent half-sample satisfies
$\E f(R)\ge M/4+\eta/4$. Returning the better of the two stored
candidates therefore gives
\begin{equation}\label{set-eq-parametric-ratio}
 \E f(S_{\mathrm{out}})
 \ge\left(\frac14+
       \frac{A_K-1/4}{1+4B_K}\right)M
 \qquad\text{provided }A_K>1/4.
\end{equation}
Indeed, as a function of $\eta/M\ge0$, the minimum of the maximum
of these two affine bounds occurs at their intersection. This
argument does not require independence between the values of the
two candidates.

We now fix the constants and the accuracy parameter. The structural
and paired-chain estimates require
$C_{\mathrm c}\ge20$ and $C_{\mathrm p}\ge24+16\ln2$;
the cohort guarantee uses $D_0=C_{\mathrm c}C_{\mathrm p}$
and $C_{\mathrm e}\ge10$. Choose
\[
 C_{\mathrm c}=20,\qquad
 C_{\mathrm p}=24+16\ln2<36,\qquad
 D_0=C_{\mathrm c}C_{\mathrm p},\qquad
 C_{\mathrm e}=10.
\]
For the accuracy parameter, take
\[
 \xi=\frac1{8C_{\mathrm e}K^2}=\frac1{80K^2}.
\]
Since $q_K-1/4\ge1/(4K)$ and $q_K\le1/2$,
\[
 A_K-\frac14
 =q_K-\frac14-\frac{q_K}{4K}
 \ge\frac1{8K}.
\]
Using $C_K\le15K+1$, $C_{\mathrm c}=20$, and $C_{\mathrm p}<36$,
we also have
\[
 1+4B_K<3000K\qquad(K\ge2).
\]
Substituting into \eqref{set-eq-parametric-ratio} gives the following
explicit form of the tradeoff.

\begin{theorem}[Maintained-set tradeoff]
\label{set-thm-tradeoff}
For every integer $K\ge2$, a fully dynamic algorithm maintains a stored
set $S_t\subseteq U_t$ with
\[
 \E f(S_t)\ge
 \left(\frac14+\frac1{24000K^2}\right)\OPT(U_t)
\]
at each fixed time of an arbitrary oblivious update sequence.
For every update prefix of length $T$, its total oracle-query cost is
\[
 O\!\left(1+n_0K^2\log^4((N+1)K)
       +TK^3N^{1/K}\log^4((N+1)K)\right)
\]
for every realization of the random choices. All query sets are
currently active, and the maintained answer can be enumerated without
additional oracle queries. Consequently, every fixed $\gamma>0$
permits a fixed advantage over $1/4$ with
$\widetilde O_\gamma(N^\gamma)$ amortized queries. For
$0<\gamma\le1$, the advantage is at least $\gamma^2/96000$.
\end{theorem}
\begin{proof}
The approximation was proved above. For the query bound, substitute
$\xi=1/(80K^2)$ and $r=\lceil N^{1/K}\rceil\le2N^{1/K}$ into
the scheduling cost. The initial cohort is charged to preprocessing.
Each later initialization is charged to actual preceding updates,
including the full cost of a module that survives beyond the reported
prefix. Thus the bound does not require a completed final epoch.

For the last assertion, choose
$K=\max\{2,\lceil1/\gamma\rceil\}$. Then $1/K\le\gamma$.
When $0<\gamma\le1$, also $K\le2/\gamma$, giving the stated
advantage. The displayed advantage has order $K^{-2}$ when the
depth is allowed to vary.
\end{proof}

\noindent{\bf AI Use Statement.} This work resulted from extensive interaction between the author and OpenAI Astra, with the author providing suggestions and critically checking Astra’s responses through iterative exchanges. OpenAI Astra was also used  in writing and revising the manuscript through detailed prompts from the author, with Gemini Argon assisting in some parts. The author directed the process and evaluated the proposed ideas and arguments. Lean files formalizing the main theorem are available on GitHub \url{https://github.com/hajiaghayi/Dynamic-Unconstrained-Submodular-Maximization}.

\begingroup
\small
\bibliographystyle{alpha}
\bibliography{fully_dynamic_unconstrained_submodular}

@inproceedings{BBG+26,
  author    = {Banihashem, Kiarash and  Goudarzi, Samira and Hajiaghayi, MohammadTaghi and Jabbarzade, Peyman and Monemizadeh, Morteza},
  title     = {When Streaming Fails: Dynamic Algorithms for Unconstrained Submodular Maximization},
  booktitle = {Advances in Neural Information Processing Systems (NeurIPS)},
  year      = {2026},
  note      = {To appear}
}

@inproceedings{FMV07,
  author = {Feige, Uriel and Mirrokni, Vahab S. and Vondr{\'a}k, Jan},
  title = {Maximizing Non-monotone Submodular Functions},
  booktitle = {Proceedings of the 48th Annual IEEE Symposium on Foundations of Computer Science (FOCS)},
  pages = {461--471},
  publisher = {IEEE Computer Society},
  year = {2007},
  doi = {10.1109/FOCS.2007.29}
}

@article{FMV11,
  author = {Feige, Uriel and Mirrokni, Vahab S. and Vondr{\'a}k, Jan},
  title = {Maximizing Non-monotone Submodular Functions},
  journal = {SIAM Journal on Computing},
  volume = {40},
  number = {4},
  pages = {1133--1153},
  year = {2011},
  doi = {10.1137/090779346}
}

@inproceedings{BFNS12,
  author = {Buchbinder, Niv and Feldman, Moran and Naor, Joseph and Schwartz, Roy},
  title = {A Tight Linear Time {$(1/2)$}-Approximation for Unconstrained Submodular Maximization},
  booktitle = {Proceedings of the 53rd Annual IEEE Symposium on Foundations of Computer Science (FOCS)},
  pages = {649--658},
  publisher = {IEEE Computer Society},
  year = {2012},
  doi = {10.1109/FOCS.2012.73}
}

@article{BFNS15,
  author = {Buchbinder, Niv and Feldman, Moran and Naor, Joseph and Schwartz, Roy},
  title = {A Tight Linear Time {$(1/2)$}-Approximation for Unconstrained Submodular Maximization},
  journal = {SIAM Journal on Computing},
  volume = {44},
  number = {5},
  pages = {1384--1402},
  year = {2015},
  doi = {10.1137/130929205}
}

@article{BF18,
  author = {Buchbinder, Niv and Feldman, Moran},
  title = {Deterministic Algorithms for Submodular Maximization Problems},
  journal = {ACM Transactions on Algorithms},
  volume = {14},
  number = {3},
  pages = {32:1--32:20},
  year = {2018},
  doi = {10.1145/3184990}
}

@inproceedings{OGV11,
  author = {{Oveis Gharan}, Shayan and Vondr{\'a}k, Jan},
  title = {Submodular Maximization by Simulated Annealing},
  booktitle = {Proceedings of the 22nd Annual ACM--SIAM Symposium on Discrete Algorithms (SODA)},
  pages = {1098--1116},
  publisher = {SIAM},
  year = {2011},
  doi = {10.1137/1.9781611973082.83}
}

@inproceedings{FNS11,
  author = {Feldman, Moran and Naor, Joseph and Schwartz, Roy},
  title = {Nonmonotone Submodular Maximization via a Structural Continuous Greedy Algorithm},
  booktitle = {Proceedings of the 38th International Colloquium on Automata, Languages, and Programming (ICALP), Part I},
  series = {Lecture Notes in Computer Science},
  volume = {6755},
  pages = {342--353},
  publisher = {Springer},
  year = {2011},
  doi = {10.1007/978-3-642-22006-7_29}
}

@inproceedings{Mon20,
  author = {Monemizadeh, Morteza},
  title = {Dynamic Submodular Maximization},
  booktitle = {Advances in Neural Information Processing Systems},
  volume = {33},
  year = {2020},
  url = {https://proceedings.neurips.cc/paper/2020/hash/6fbd841e2e4b2938351a4f9b68f12e6b-Abstract.html}
}

@inproceedings{LMNTZ20,
  author = {Lattanzi, Silvio and Mitrovi{\'c}, Slobodan and Norouzi-Fard, Ashkan and Tarnawski, Jakub and Zadimoghaddam, Morteza},
  title = {Fully Dynamic Algorithm for Constrained Submodular Optimization},
  booktitle = {Advances in Neural Information Processing Systems},
  volume = {33},
  year = {2020},
  url = {https://proceedings.neurips.cc/paper_files/paper/2020/hash/9715d04413f296eaf3c30c47cec3daa6-Abstract.html}
}

@inproceedings{BBGHJM23a,
  author = {Banihashem, Kiarash and Biabani, Leyla and Goudarzi, Samira and Hajiaghayi, MohammadTaghi and Jabbarzade, Peyman and Monemizadeh, Morteza},
  title = {Dynamic Constrained Submodular Optimization with Polylogarithmic Update Time},
  booktitle = {Proceedings of the 40th International Conference on Machine Learning (ICML)},
  series = {Proceedings of Machine Learning Research},
  volume = {202},
  pages = {1660--1691},
  publisher = {PMLR},
  year = {2023},
  url = {https://proceedings.mlr.press/v202/banihashem23a.html}
}

@inproceedings{BBGHJM23b,
  author = {Banihashem, Kiarash and Biabani, Leyla and Goudarzi, Samira and Hajiaghayi, MohammadTaghi and Jabbarzade, Peyman and Monemizadeh, Morteza},
  title = {Dynamic Non-monotone Submodular Maximization},
  booktitle = {Advances in Neural Information Processing Systems},
  volume = {36},
  year = {2023},
  doi = {10.52202/075280-0761},
  url = {https://proceedings.neurips.cc/paper_files/paper/2023/hash/387982dbf23d9975c7fc45813dd3dabc-Abstract-Conference.html}
}

@inproceedings{DFLNZ23,
  author = {D{\"u}tting, Paul and Fusco, Federico and Lattanzi, Silvio and Norouzi-Fard, Ashkan and Zadimoghaddam, Morteza},
  title = {Fully Dynamic Submodular Maximization over Matroids},
  booktitle = {Proceedings of the 40th International Conference on Machine Learning (ICML)},
  series = {Proceedings of Machine Learning Research},
  volume = {202},
  pages = {8821--8835},
  publisher = {PMLR},
  year = {2023},
  url = {https://proceedings.mlr.press/v202/duetting23a.html}
}

@inproceedings{CP22,
  author = {Chen, Xi and Peng, Binghui},
  title = {On the Complexity of Dynamic Submodular Maximization},
  booktitle = {Proceedings of the 54th Annual ACM SIGACT Symposium on Theory of Computing (STOC)},
  pages = {1685--1698},
  publisher = {ACM},
  year = {2022},
  doi = {10.1145/3519935.3519951},
  url = {https://arxiv.org/abs/2111.03198}
}

@inproceedings{BBGHJM24,
  author = {Banihashem, Kiarash and Biabani, Leyla and Goudarzi, Samira and Hajiaghayi, MohammadTaghi and Jabbarzade, Peyman and Monemizadeh, Morteza},
  title = {Dynamic Algorithms for Matroid Submodular Maximization},
  booktitle = {Proceedings of the 35th Annual ACM--SIAM Symposium on Discrete Algorithms (SODA)},
  pages = {3485--3533},
  publisher = {SIAM},
  year = {2024},
  doi = {10.1137/1.9781611977912.125},
  url = {https://arxiv.org/abs/2306.00959}
}

@inproceedings{BBGHJM25,
  author = {Banihashem, Kiarash and Biabani, Leyla and Goudarzi, Samira and Hajiaghayi, MohammadTaghi and Jabbarzade, Peyman and Monemizadeh, Morteza},
  title = {Dynamic Algorithms for Submodular Matching},
  booktitle = {Proceedings of the 52nd International Colloquium on Automata, Languages, and Programming (ICALP)},
  series = {Leibniz International Proceedings in Informatics},
  volume = {334},
  pages = {19:1--19:21},
  publisher = {Schloss Dagstuhl -- Leibniz-Zentrum f{\"u}r Informatik},
  year = {2025},
  doi = {10.4230/LIPIcs.ICALP.2025.19},
  url = {https://drops.dagstuhl.de/entities/document/10.4230/LIPIcs.ICALP.2025.19}
}

@inproceedings{BGHJM25,
  author = {Banihashem, Kiarash and Goudarzi, Samira and Hajiaghayi, MohammadTaghi and Jabbarzade, Peyman and Monemizadeh, Morteza},
  title = {Non-monotone Submodular Optimization: {$p$}-Matchoid Constraints and Fully Dynamic Setting},
  booktitle = {Advances in Neural Information Processing Systems},
  volume = {38},
  year = {2025},
  doi = {10.52202/085713-3471},
  url = {https://proceedings.neurips.cc/paper_files/paper/2025/hash/963acdcdf87fec84916b10949cbee255-Abstract-Conference.html}
}

@article{GZZ26,
  author = {Gao, Yiwei and Zhang, Jialin and Zhang, Zhijie},
  title = {Improved Fully Dynamic Submodular Maximization Under Matroid Constraints},
  journal = {Proceedings of the AAAI Conference on Artificial Intelligence},
  volume = {40},
  number = {43},
  pages = {36928--36936},
  year = {2026},
  doi = {10.1609/aaai.v40i43.41020},
  url = {https://ojs.aaai.org/index.php/AAAI/article/view/41020}
}

@inproceedings{BGHJM26,
  author = {Banihashem, Kiarash and Goudarzi, Samira and Hajiaghayi, MohammadTaghi and Jabbarzade, Peyman and Monemizadeh, Morteza},
  title = {Improved Dynamic Algorithm for Non-monotone Submodular Maximization under Cardinality Constraint},
  booktitle = {Proceedings of the 43rd International Conference on Machine Learning (ICML)},
  series = {Proceedings of Machine Learning Research},
  volume = {306},
  pages = {6298--6308},
  publisher = {PMLR},
  year = {2026},
  url = {https://proceedings.mlr.press/v306/banihashem26c.html}
}

@inproceedings{BMKK14,
  author = {Badanidiyuru, Ashwinkumar and Mirzasoleiman, Baharan and Karbasi, Amin and Krause, Andreas},
  title = {Streaming Submodular Maximization: Massive Data Summarization on the Fly},
  booktitle = {Proceedings of the 20th ACM SIGKDD International Conference on Knowledge Discovery and Data Mining (KDD)},
  pages = {671--680},
  publisher = {ACM},
  year = {2014},
  doi = {10.1145/2623330.2623637}
}

@inproceedings{LFKK22,
  author = {Li, Wenxin and Feldman, Moran and Kazemi, Ehsan and Karbasi, Amin},
  title = {Submodular Maximization in Clean Linear Time},
  booktitle = {Advances in Neural Information Processing Systems},
  volume = {35},
  pages = {17473--17487},
  year = {2022},
  doi = {10.52202/068431-1270},
  url = {https://proceedings.neurips.cc/paper_files/paper/2022/hash/6faf3b8ed0df532c14d0fc009e451b6d-Abstract-Conference.html}
}

@inproceedings{PR23,
  author = {Peng, Binghui and Rubinstein, Aviad},
  title = {Fully-Dynamic-to-Incremental Reductions with Known Deletion Order (e.g. Sliding Window)},
  booktitle = {Proceedings of the 2023 Symposium on Simplicity in Algorithms (SOSA)},
  pages = {261--271},
  publisher = {SIAM},
  year = {2023},
  doi = {10.1137/1.9781611977585.ch24}
}

@inproceedings{LS24,
  author = {Liu, Quanquan C. and Srinivas, Vaidehi},
  title = {The Predicted-Updates Dynamic Model: Offline, Incremental, and Decremental to Fully Dynamic Transformations},
  booktitle = {Proceedings of the 37th Conference on Learning Theory (COLT)},
  series = {Proceedings of Machine Learning Research},
  volume = {247},
  pages = {3582--3641},
  publisher = {PMLR},
  year = {2024},
  url = {https://proceedings.mlr.press/v247/liu24c.html}
}

@inproceedings{AB24,
  author = {Agarwal, Arpit and Balkanski, Eric},
  title = {Learning-Augmented Dynamic Submodular Maximization},
  booktitle = {Advances in Neural Information Processing Systems},
  volume = {37},
  pages = {14148--14176},
  year = {2024},
  doi = {10.52202/079017-0453},
  url = {https://proceedings.neurips.cc/paper_files/paper/2024/hash/19cdab1dee61d55158cf106244ceab08-Abstract-Conference.html}
}

@incollection{McDiarmid89,
  author = {McDiarmid, Colin},
  title = {On the Method of Bounded Differences},
  booktitle = {Surveys in Combinatorics, 1989},
  editor = {Siemons, J.},
  series = {London Mathematical Society Lecture Note Series},
  volume = {141},
  pages = {148--188},
  publisher = {Cambridge University Press},
  year = {1989},
  doi = {10.1017/CBO9781107359949.008}
}

@inproceedings{IJB13,
  author = {Iyer, Rishabh and Jegelka, Stefanie and Bilmes, Jeff},
  title = {Fast Semidifferential-based Submodular Function Optimization},
  booktitle = {Proceedings of the 30th International Conference on Machine Learning (ICML)},
  series = {Proceedings of Machine Learning Research},
  volume = {28},
  note = {Part 3},
  pages = {855--863},
  publisher = {PMLR},
  year = {2013},
  url = {https://proceedings.mlr.press/v28/iyer13.html}
}

@incollection{Edmonds70,
  author = {Edmonds, Jack},
  title = {Submodular Functions, Matroids, and Certain Polyhedra},
  booktitle = {Combinatorial Structures and Their Applications},
  pages = {69--87},
  publisher = {Gordon and Breach},
  year = {1970}
}

@article{CCPV11,
  author = {C{\u a}linescu, Gruia and Chekuri, Chandra and P{\'a}l, Martin and Vondr{\'a}k, Jan},
  title = {Maximizing a Monotone Submodular Function Subject to a Matroid Constraint},
  journal = {SIAM Journal on Computing},
  volume = {40},
  number = {6},
  pages = {1740--1766},
  year = {2011},
  doi = {10.1137/080733991}
}
\endgroup
\end{document}